\documentclass[10pt]{article}
\usepackage{style}
\usepackage{subcaption}
\usepackage{qcircuit}
\pdfoutput=1

\AtEveryBibitem{ %
    \clearfield{day}
    \clearfield{month}
    \clearfield{series}
    \clearfield{venue}
    \clearname{editor}
    \clearlist{publisher}
    \clearlist{location} %
    \clearfield{venue}
    \clearfield{issn}
    \clearfield{isbn}
    \clearfield{urldate}
    \clearfield{eventdate}
    \clearfield{pages}
    \clearfield{number}
    \clearfield{volume}
}

\newcommand{\dec}{\mathrm{Dec}}
\newcommand{\supp}{\mathrm{supp}}
\newcommand{\punc}{\mathrm{punc}}
\newcommand{\cl}{\mathrm{cl}}
\newcommand{\Span}{\operatorname{span}}
\newcommand{\diag}{\mathrm{diag}}

\newcommand{\dist}{\mathrm{D}}
\newcommand{\col}{\mathrm{col}}
\newcommand{\im}{\mathrm{im}}
\newcommand{\rank}{\mathrm{rk}}
\newcommand{\RM}{\mathrm{RM}}

\newcommand{\eval}{\mathrm{eval}}

\allowdisplaybreaks

\newcommand{\NEXP}{\mathsf{NEXP}} 
\newcommand{\MIP}{\mathsf{MIP}} 
\newcommand{\QMA}{\mathsf{QMA}}
\newcommand{\be}{\boldsymbol{e}}

\title{No low-degree tests for quantum states}
\author{Omar Alrabiah\thanks{UC Berkeley. \texttt{\{oalrabiah,jswright\}@berkeley.edu}} \and Srinivasan Arunachalam\thanks{IBM Quantum. \texttt{srinivasan.arunachalam@ibm.com}} \and Sabee Grewal\thanks{UT Austin. \texttt{sabee@cs.utexas.edu}} 
\and John Wright\footnotemark[1]}

\date{}

\begin{document}

\maketitle

\begin{abstract}
We study the problem of testing \emph{low-degree phase states}, namely $m$-qudit quantum states of the form $q^{-m/2}\sum_{x \in \F_q^m} \omega^{f(x)} \ket{x}$, where $f$ is a degree-$d$ polynomial. 
In contrast to the classical setting, where low-degree polynomials admit highly efficient classical testers, it is not known whether analogous quantum tests exist.
We show that no such quantum low-degree test exists:
any tester requires $\Omega(\binom{\lfloor m/2\rfloor}{\lfloor (d-1)/2 \rfloor})$ copies to determine whether a given state is a degree-$d$ phase state or is far from every such state. 
Our results follow from a general framework that relates quantum testing of \emph{codeword phase states} to classical decoding properties of the dual code, which allows us to leverage known bounds on the tolerance of high-rate Reed--Muller codes to random errors.
\end{abstract}

\hypersetup{linktocpage}
\setcounter{tocdepth}{2}
{\tableofcontents}
\setlength{\parskip}{4pt}%
\newpage
\section{Introduction}

One of the recurring themes throughout theoretical computer science is the surprising power of low-degree polynomials.
They form the basis of Reed--Solomon and Reed--Muller codes, two of the most widely studied error-correcting codes in the field of coding theory.
They also lie at the heart of the original ``algebraic'' proof of the PCP theorem~\cite{AS98,ALM+98},
one of the central results in modern-day complexity theory.
One of their key properties is that they are \emph{locally testable}~\cite{alon2005testing,jutla2004testing,kaufman2006testing,bhattacharyya2010optimal},
meaning that given a function $f:\F_q^m \rightarrow \F_q$,
one can test whether $f$ is a degree-$d$ polynomial by querying it on an extremely small number of points.
Indeed, this ``low-degree test'' is quite natural: simply choose $d+2$ random collinear points $x_1, \ldots, x_{d+2} \in \F_q^m$ and check that the values of $f$ on these points $f(x_1), \ldots, f(x_{d+2})$ are consistent with some degree-$d$ polynomial.

In this work, we study quantum analogues of low-degree polynomials,
i.e., states which encode low-degree polynomials in their phases.
To define these, let $q$ be a prime number,
    and set $\omega = e^{2 \pi i/q}$.
    Let $f: \F_q^m \rightarrow \F_q$ be a degree-$d$ polynomial.
    Then the corresponding \emph{degree-$d$ phase state} is given by
    \begin{equation*}
        \ket{\psi_f} \coloneq \frac{1}{\sqrt{q^m}} \sum_{x \in \F_q^m} \omega^{f(x)} \cdot \ket{x}.
    \end{equation*}
Phase states are ubiquitous throughout quantum information, with applications in measurement-based quantum computation~\cite{hein2004multiparty},  algorithms for the hidden subgroup problem~\cite{decker2007efficient,DBLP:journals/cjtcs/BaconCD06}, cryptography~\cite{ji-pseudorandom-states2018,brakerski10.1007/978-3-030-36030-6_10}, quantum advantage proposals~\cite{montanaro2016quantum,bremner2016average}, and complexity theory~\cite{DBLP:conf/coco/IraniNNRY22}.

The problem of \emph{learning} low-degree phase states is one of the oldest topics in quantum information.
For example, the classic Bernstein--Vazirani algorithm~\cite{BV97} allows us to learn degree-1 phase states with just a single copy of the state.
Later, R\"{o}tteler~\cite{Rot09} gave an algorithm for learning degree-2 phase states over $\F_2$ with just $O(m)$ copies of the state.
More recently, Arunachalam, Bravyi, Dutt, and Yoder~\cite{arunachalam2022phase}  characterized the number of copies needed to learn a degree-$d$ phase state over $\F_2$, showing that it is exactly $\Theta(m^{d-1})$.
(We note that the learnability over $\F_q$ remains~open.)

We are interested in the related problem of \emph{testing} low-degree phase states.
In this setting, one is given copies of a quantum state $\ket{\psi}$ and asked to determine whether it is (YES) a degree-$d$ phase state or (NO) $\epsilon$-far from all degree-$d$ phase states.
Learning is typically more difficult than testing, and the above learning algorithms immediately imply that degree-$d$ phase states over $\F_2$ are testable with $O(m^{d-1}/\epsilon^2)$ copies.
Over larger fields $\F_q$,
note that because there are at most $a = m^d$ monomials over $\F_q^m$ of degree at most $d$,
the number of degree-$d$ polynomials $f$
(and therefore the number of degree-$d$ phase states $\ket{\psi_f}$)
is at most $q^{a}$.
As a result, a theorem of Harrow, Lin, and Montanaro~\cite{HLM17}  implies that the set of degree-$d$ phase states is testable with $O(\mathrm{log}(q^{a})/\epsilon^2) = O(m^d \log(q)/\epsilon^2)$ copies (see~\cite{HKP20} for an alternative proof of this fact).
Note that these bounds are all exponential in $d$,
which is significantly worse than the $\poly(d)$ bound the classical low-degree test suggests that we might hope for.
This motivates the following question (explicitly raised in~\cite{anshu2023survey,caro2024testing}): 
\begin{center}
    \textit{
    Is there a quantum analogue of the low-degree test\\
    which allows us to test low-degree phase states with extreme efficiency?
    }
\end{center}
Aside from being a natural  question in its own right,
we will see below that a resolution of this question in the positive direction
would also have dramatic consequences for the study of quantum complexity theory.
For this application, we will see that it would suffice for degree-$d$ phase states to be testable with $\poly(m, d, q)$ copies.
There is very little evidence in the literature pointing in either direction on this question,
but there is one piece of evidence which leans slightly positive: stabilizer states, which are closely related to degree-2 phase states, can be tested with only 6 copies~\cite{GNW21}.

\paragraph{Quantum complexity theory.}
One of our primary motivations comes from the study of the quantum complexity class $\QMA(2)$.
$\QMA(2)$ is the set of all languages that can be decided by a verifier who has access to two proof states $\ket{\pi_1}$ and $\ket{\pi_2}$ that are promised to be unentangled.
The study of $\QMA(2)$ is deeply tied to our understanding of entanglement and separability,
and as a result, determining its power is one of the major open problems in quantum complexity theory.
To date, all that is known about $\QMA(2)$ and its relationships to traditional complexity classes is that it contains $\QMA$ and is contained in $\NEXP$, both of which are essentially trivial containments, and there is very little evidence suggesting where exactly $\QMA(2)$ should fall in between these two extremes.

The recent complexity theoretic result $\MIP^* = \mathsf{RE}$~\cite{JNV+20}, however,
suggests that strange quantum resources can sometimes yield surprising amounts of computational power. 
Might, therefore, $\QMA(2)$ actually be equal to $\NEXP$?
One piece of evidence that this is not the case comes from the work of Gharibian, Santha, Sikora, Sundaram, and Yirka~\cite{GSS+22}, who have shown that $\QMA(2) \subseteq \mathsf{Q\Sigma}_3 \subseteq \NEXP$, where $\mathsf{Q\Sigma}_3$ is the third level of the ``quantum proof polynomial hierarchy'' that they define.
A proof that $\QMA(2) = \NEXP$ would then imply that $\QMA(2) = \mathsf{Q\Sigma}_3$, which would indicate a partial collapse in their polynomial hierarchy.
Given the relative recency of these complexity classes, however, it is still unclear how strong of a barrier this represents, and we believe that $\QMA(2) = \NEXP$ remains a tantalizing possibility.

Perhaps the most natural approach for showing $\QMA(2) = \NEXP$ is to begin with the PCP characterization of $\NEXP$,
which states that any language in $\NEXP$ can be decided by a polynomial-time classical verifier with oracle access to an exponentially long proof, which it queries at only polynomially many locations (in fact, $O(1)$ queries suffice, though we will not need this here).
More concretely, the canonical $\NEXP$-complete language is $\mathsf{Succinct}$-$\mathsf{3Sat}$,
in which one is given a $\mathsf{3Sat}$ formula $\varphi$ on $N = 2^n$ variables $x = (x_1, \ldots, x_N)$, represented succinctly as a $\poly(n)$-size Boolean circuit,
and the goal is to determine whether $\varphi$ is satisfiable.
Even if $\varphi$ is satisfiable,
a satisfying assignment $x$ is simply too long for the $\poly(n)$-time verifier to read,
and so instead of being provided $x$ as a proof, the proof the verifier is given is a function  $f:\F_q^m \rightarrow \F_q$ which should be a degree-$d$ encoding of $x$ for some ``small'' value of $d$.
Note that we will need to pick parameters so that the number of degree-$d$ polynomials $f$ is at least the number of assignments to the $N$ variables of $\varphi$ (otherwise, two different assignments will get mapped to the same encoding),
 so we will need $q^a \geq 2^N$, where $a = m^d$;
a typical setting of parameters in this literature takes
\begin{equation}\label{eq:parameters}
    m = \Theta\Big(\frac{n}{\log(n)}\Big),\quad
    q = \poly(n),\quad
    d = \Theta\Big(\frac{n^2}{\log(n)}\Big).
\end{equation}
Then the verifier, given access to the proof $f$,
will run a low-degree tester to query $f$ on $d+2 = \poly(n)$ locations to check that it is indeed a degree-$d$ polynomial.
If it passes this test, the verifier will perform some further checks to certify that it encodes a satisfying assignment of $\varphi$.

To port this over to the $\QMA(2)$ setting,
rather than providing the quantum verifier access to the exponentially-long proof $f$,
we can imagine providing it access to a proof state $\ket{\psi}$, which is meant to be equal to $\ket{\psi_f}$ for some degree-$d$ polynomial $f: \F_q^m \rightarrow \F_q$.
Note that this state can be stored on $m \log(q) = \poly(n)$ qubits, and so it is a compact representation of an exponentially-long proof.
Using the characterization $\QMA(2) = \QMA(\mathrm{poly}(n))$ due to Harrow and Montanaro~\cite{HM13},
which states that two unentangled proofs are equivalent in power to polynomially many proofs,
we can provide the verifier with $\poly(n)$ copies of~$\ket{\psi}$ rather than just two.
To simulate the PCP verifier,
the verifier should begin by using these $\poly(n)$ copies to test if $\ket{\psi}$ is a degree-$d$ phase state,
and whether it can do this is exactly the question of whether there exists a hyperefficient low-degree test for phase states.
Note that this application suggests how efficient we would like our test to be:
since we want a $\poly(n)$-copy tester,
and $m$, $q$, and $d$ are all $\poly(n)$,
this suggests that our tester should consume a number of copies polynomial in $m$, $q$, and $d$.

Performing the low-degree test is just the first step of the PCP protocol for $\NEXP$,
and to show $\QMA(2) = \NEXP$ we would need to simulate the remaining steps in $\QMA(2)$ as well.
However, this turns out not to be an issue after all: we have been able to show that \emph{if} there exists a tester for degree-$d$ phase states that uses $\poly(m, q, d)$ copies, and \emph{if} this tester is also computationally efficient,
then $\QMA(2)$ does indeed equal $\NEXP$.
We omit the proof here, though we note that it follows similar proofs in the literature, such as \cite{Raz05-qipcp,aaronson2024pdqmadqma,grewal2025unentanglementpostmeasurementbranchingquantum}.
We believe that this gives strong motivation for studying the problem of testing low-degree phase states.

\subsection{Our results}

For our main result, we show that there are no hyperefficient low-degree tests for quantum states, for a wide range of parameters.

\begin{theorem}[No low-degree tests for quantum states]\label{thm:main}
Given copies of an $m$-qudit state $\ket{\psi}$, any algorithm that distinguishes whether it's a degree-$d$ phase state or $\eps$-far from all such states requires
\[
\Omega\left(\binom{\lfloor m/2 \rfloor}{\lfloor (d-1)/2\rfloor} \right) 
\]
copies. For $d=2$, $\Omega(m)$ copies are necessary.
\end{theorem}

The bounds in \cref{thm:main} hold both when the field size $q$ is fixed or when it grows with the number of variables $m$; in fact, our lower bounds become stronger when $q$ increases, and apply for all $q = \poly(m)$. In particular, \cref{thm:main} rules out any tester using $\poly(m,d,q)$ copies. Consequently, there is no efficient quantum-state low-degree test for the parameter regime specified in \cref{eq:parameters}, and hence the natural approach of simulating the $\NEXP$ PCP verifier within $\QMA(2)$ via such tests cannot succeed.

Our result highlights a distinct difference between the classical and quantum settings. 
In the classical setting, we are given query access to a low-degree function $f$, and can use this to access points which are related to each other, namely collinear points $x_1, \ldots, x_{d+2}$.
On the other hand, low-degree phase states are uniform superpositions over all points $x \in \F_q^m$, and $t$ copies of a phase state are $t$ independently uniform superpositions.
This means that it should be difficult to extract information from copies of the state that relies on highly correlated subsets of points.
Our main theorem confirms that this is the case, and shows that there is no ``quantum magic'' for getting around this barrier.

Our results suggest that the correct classical analogue to testing low-degree phase states might be the \emph{passive testing} scenario rather than the usual property testing scenario.
In the passive testing scenario,
rather than being allowed to query the unknown function $f : \F_q^m \rightarrow \F_q$ at any point of one's choosing,
one is given independent samples of the form $(\bx, f(\bx))$, where $\bx$ is drawn uniformly at random from $\F_q^m$.
The work of \cite{AHW16} showed that $\Theta(\binom{m}{\leq d})$ passive samples are both necessary and sufficient for passive learning and testing of degree-$d$ polynomials.
For small $d$, note that $\binom{m}{\leq d} \approx m^d$,
which matches our bound in the ``small $d$'' regime,
except with an additional factor of $m$.
(Note that we expect a factor $m$ savings since Bernstein-Vazirani allows us to learn and test degree-1 phase states with a single copy.)

\subsection{Technical overview}

Our \Cref{thm:main} actually follows as a special case of a much more general framework that allows us to prove lower bounds for the problem of testing phase states drawn from arbitrary linear codes.
We define these more general codeword phase states as follows.

\begin{definition}[Codeword states]
    Let $q$ be a prime number and set $\omega \coloneqq e^{2 \pi i/q}$.
    Given a vector $x \in \F_q^n$, the corresponding phase state is given by
    \begin{equation*}
        \ket{\psi_x} \coloneqq \frac{1}{\sqrt{n}} \sum_{i=1}^n \omega^{x_i}  \ket{i}.
    \end{equation*}
    Let $C \subseteq \F_q^n$ be a linear code.
    For a codeword $c \in C$, the corresponding  \emph{$C$-codeword phase state}~is~$\ket{\psi_c}$.
\end{definition}
\noindent
Degree-$d$ phase states arise as a special case by taking $C = \RM_q[m,d]$, the $m$-variate degree-$d$ Reed--Muller code over $\F_q$.  

For a general code $C$, we consider the problem of testing whether an unknown state $\ket{\psi}$ is (YES) a $C$-codeword phase state or (NO) is $\epsilon$-far from any $C$-codeword phase state.
Our framework establishes a connection between 
the testability of $C$-codeword states and the decodability properties of the dual code $C^{\perp}$.
In particular, we care about the decodability of the dual code in the following~sense.

\begin{definition}[Average-case decoding]
    Let $C \subseteq \F_q^n$ be a linear code of dimension $k$. Let $H \in \F_q^{(n-k) \times n}$ be its parity check matrix. Let $\be_t \in \F_q^n$ be a random ``error'' distributed as follows:
    \begin{equation*}
        \be_t
        = \delta_{\ba_1} + \cdots + \delta_{\ba_t},
    \end{equation*}
    where $\ba_1, \ldots, \ba_t$ are sampled independently and uniformly at random from $[n]$, and $\delta_i$ is the vector with $1$ in the $i$'th entry and $0$'s elsewhere.
    \emph{Average-case decoding with respect to the random error $\be_t$} refers to the following scenario involving a ``decoder'':
    \begin{enumerate}
        \item A random error $\be_t$ is sampled.
        \item The decoder is given the syndrome $\bs = H \cdot \be_t$.
        \item The decoder outputs a guess $\bg \in \F_q^n$ for the error.
        \item The decoder succeeds if $\bg = \be_t$.
    \end{enumerate}
    Define $p_{\mathrm{decode}}^t(C)$ to be the maximum probability of the decoder succeeding in this scenario.
\end{definition}

Note that there is a clear optimal procedure for the decoder to follow, which is to always guess the vector $\bg \in \F_q^n$ which satisfies $H \cdot \bg = \bs$ and has the highest probability of being equal to the error $\be_t$.
This is known as the \emph{maximum a posteriori (MAP)} decoder,
and it is easy to see that it has a successful decoding probability of
\begin{equation*}
    p_{\mathrm{decode}}^t(C) = \sum_{s \in \F_q^{n-k}} \max_{g:H \cdot g = s}\big\{\Pr[\be_t=g]\big\}.
\end{equation*}
Average-case decoding is more commonly described in the following equivalent manner: a codeword $c \in C$ has been corrupted by an error $\be_t$, resulting in the noisy codeword $\bc' = c + \be_t$, and the decoder is tasked with recovering $c$.
To do so, the decoder applies the parity check matrix to extract the~syndrome
\begin{equation*}
\bs
= H \cdot \bc'
= H \cdot c + H \cdot \be_t
= H \cdot \be_t,
\end{equation*}
and they use this to compute a guess $\bg$ for the most likely error which has occurred.
Given this, they will decode the noisy codeword $c'$ to the codeword $c' - \bg = c + \be_t - \bg$, which is correct if $\bg = \be_t$.

With this in place, we can state our framework for general codes. Our goal is to show lower bounds for testing $C$-codeword states. We will do so by showing lower bounds on the problem of distinguishing $t$ copies of (YES) a random $C$-codeword phase state $\ket{\psi_{\bc}}$ and (NO) a uniformly random phase state $\ket{\psi_\bx}$, where $\bx \sim \F_q^n$.
This implies a lower bound on testing $C$-codeword states so long as $C$ is not too dense,
as a random $\ket{\psi_\bx}$ should be very far from any $C$-codeword state with very high probability.
This motivates the following~notation.
\begin{notation}[Average codeword state]
    Let $C \subseteq \F_q^n$ be a linear code, and let $t$ be an integer. Then we define the mixed state
    \begin{equation*}
        \rho_C^t \coloneqq \E_{\bc \sim C}[\ketbra{\psi_{\bc}}{\psi_{\bc}}^{\otimes t}].
    \end{equation*}
    When $C$ is all of $\F_q^n$, we write $\rho^t_{\mathrm{all}} \coloneqq \rho^t_{\F_q^n}$.
\end{notation}

Our goal, then, is to show that $\rho^t_C$ and $\rho^t_{\mathrm{all}}$ are close in trace distance, which implies that they are information-theoretically indistinguishable.
Our next lemma shows that this is the case so long as the dual code $C^{\perp}$ can be decoded with probability close to 1.

\begin{lemma}[Decodability of the dual code implies no $C$-codeword tester]
\label{lemma:intro-decode}
    For any linear code $C \subseteq \F_q^n$, we~have
    \begin{equation*}
        \mathrm{D}_{\mathrm{tr}}(\rho_C^t, \rho_{\mathrm{all}}^t) \leq \sqrt{1 - p_{\mathrm{decode}}^t(C^{\perp})^2}.
    \end{equation*}
\end{lemma}

This translates the question of showing lower bounds for testing $C$-codeword phase states into a purely classical question about the average-case decodability of the dual code $C^\perp$ with respect to the error $\be_t$.
Average-case decoding over different noise channels is one of the most basic topics in all of coding theory, dating back to Shannon's pioneering work in the 40's~\cite{Sha48},
and there is a wealth of literature for us to draw on.
For our bounds, we care about the case when our code $C$ is the Reed--Muller code $\RM_q[m, d]$,
and it is well-known that the corresponding dual code is the Reed--Muller code $C^{\perp} = \RM_q[m, m(q-1)-d-1]$.
Average-case decodability of the Reed--Muller code is a particularly well-studied topic, and strong bounds for (a closely related distribution to) our error distribution have been shown by Abbe, Shpilka, and Wigderson~\cite{ASW15} for the field $\F_2$.
Proving our main result requires us to generalize their bound to the field $\F_q$.
To state our bound, let us write $D_{d,q}(r)$ for the number of monomials in $r$ variables of total degree at most $d$, with each variable having degree at most $q-1$;
note that this is simply the dimension of the code $\RM_q[r,d]$.
Then we show the following generalization of {\cite[Theorem 6.2]{ASW15}.
\begin{theorem}[Average-case decodability of the Reed--Muller code]\label{thm:our-generalization}
    Let $q$ be a prime number.
    Let $C = \RM_q[m, 2d+1]$.
    Let $0 \leq r \leq m$,
    and let $t \leq D_{d,q}(r)$ be a positive integer.
Then
\begin{equation*}
    p^t_{\mathrm{decode}}(C^{\perp}) \geq 1 - t q^{r-m}.
\end{equation*}
\end{theorem}

Let us note a couple of peculiarities of this theorem statement.
First, it is stated in terms of $\RM_q[m, 2d+1]$ rather than the code that we care about, which is $\RM_q[m, d]$; this is because if we state it for $\RM_q[m, d]$, then we will need to refer to the quantity $D_{(d-1)/2, q}(r)$ and ensure that $d$ is odd, which is mildly cumbersome to deal with.
Next, note that $r$ is a free parameter, and must be tuned to achieve the strongest bound possible.

We also note that this framework yields lower bounds for degree-$d$ phase states with $d \ge 3$.
For degree-$2$ phase states, we prove an $\Omega(m)$ lower bound via a different approach. Specifically, we show that distinguishing a random degree-$2$ phase state from one restricted to a halfspace requires $\Omega(m)$ copies.

While our primary focus is on property testing lower bounds, our framework yields a stronger consequence. Namely, whenever $p_{\rm decode}^t(C^\perp)$ is large, the $t$-th moment operator of a random $C$-codeword phase state is close to that of a Haar-random state.
Specifically, let $\rho_{\rm Haar}^t$ denote the $t$-th moment operator of the Haar measure on pure states. It is known that $\rho_{\rm all}^t$ is close to $\rho_{\rm Haar}^t$ (see \cref{appendix:random-fq}). Combined with \cref{lemma:intro-decode}, this implies that $\rho_C^t$ is close to $\rho_{\rm Haar}^t$. Equivalently, random $C$-codeword phase states form approximate state designs.
We defer the details of this connection to \cref{sec:testing-LBs}. Here, we highlight one consequence: low-degree phase states form highly accurate designs for moments far larger than their degree.

\begin{theorem}
    Let $q$ be a prime number. 
    Define $M \coloneqq \binom{\lfloor m/2 \rfloor}{\lfloor (d-1)/2\rfloor}$ and suppose that $M = 2^{o(m)}$. Then for any $a < \frac{1}{4} \log_2 q$, the uniform ensemble of $\RM_q[m,d]$-codeword phase states is a $2^{-am}$-approximate state $t$-design for $t = \Omega(M)$.
\end{theorem}

Finally, we also show that the learnability of codeword states admits a purely classical characterization in terms of average-case decoding. 
In particular, the sample complexity is governed by the order-$1/2$ R\'enyi entropy, which we denote by $H_{1/2}(\cdot)$, of the syndrome distribution of the dual code under the error $\be_t$. 

\begin{theorem}[Learning codeword states]
\label{thm:learning}
Let $C \subseteq \F_q^n$ be a linear code, and let $H$ be a parity check matrix of $C^\perp$. 
Given $t$ copies of an unknown codeword state $\ket{\psi_c}$ for $c \in C$, there exists a learning algorithm that identifies $c$ with probability at least $1-\delta$ if and only if 
\[
H_{1/2}(H\be_t)\ge \dim(C) + \log_q(1-\delta).
\]
\end{theorem}
Thus, the copy complexity of learning codeword states reduces to the classical question of determining the smallest $t$ for which the inequality in \cref{thm:learning} is satisfied.

\subsection{Future directions}

Our work leaves open several interesting problems related to the learnability and testability of low-degree states.
First, as we have seen above, the work of Arunachalam, Bravyi, Dutt, and Yoder~\cite{arunachalam2022phase} showed that $\Theta(m^{d-1})$ copies are necessary and sufficient for learning a degree-$d$ phase state over $\F_2^m$.
But what about over $\F_q^m$?
We conjecture that these states can be learned with $\Theta(m^{d-1})$ copies as well. 

Next, we have shown lower bounds for the problem of testing degree-$d$ phase states, but we do not have matching upper bounds.
The one exception is the case of testing degree-$2$ phase states over $\F_2^m$, where our lower bound of $\Omega(m)$ matches the upper-bound of $O(m/\epsilon^2)$ that follows from the $O(m)$-copy learning algorithm of R\"{o}tteler~\cite{Rot09}.
We conjecture that this pattern holds more generally,
and that $\Theta(m^{d-1}/\epsilon^2)$ copies are necessary and sufficient to test degree-$d$ phase states over $\F_q^m$,
matching our conjectured bound for learning degree-$d$ phase states in terms of the $n$ and $d$ parameters.
This would parallel the learning and testing bounds for degree-$d$ polynomials in the passive learning and testing scenarios, with the factor-$m$ savings due to the Bernstein--Vazirani algorithm.
By showing improved error tolerance of high-rate binary Reed--Muller codes over the binary symmetric channel (BSC), one can hope to upgrade~\cref{thm:our-generalization} to show a lower bound of $\Omega(m^{d-2})$ copies,\footnote{For $d = 2$, it turns out that one cannot hope for a bound better than $O(1)$ via this approach.} nearly matching the upper bound. This belief is consistent with a vibrant line of work showing that, for various parameter regimes, the binary Reed--Muller code achieves capacity on the binary erasure channel (BEC), binary symmetric channel (BSC), and more generally, binary memoryless channels (BMS) (see the book~\cite{ASSY23} for a recap of the numerous such results). As such, one should expect the Reed--Muller code to achieve a similar error tolerance as high-rate random binary linear codes, suggesting an improved version of~\cref{thm:our-generalization} that improves $\RM_q[m,2d+1]$ to $\RM_q[m,d+2]$.

More broadly, one can ask whether there exists any linear code $C$ for which the corresponding $C$-codeword states can be tested significantly more efficiently than they can be learned.
For Reed--Muller codes of degree $1$ and $2$, the testing and learning complexities match.
We conjecture that this holds generally for all linear codes.

\begin{conjecture}[No hyperefficient $C$-codeword testers]
For every linear code $C$, the copy complexity of testing and learning $C$-codeword states is the same.
\end{conjecture}

Finally, although our lower bounds rule out the use of low-degree phase states for proving $\QMA(2) = \NEXP$,  they don't rule out using some \emph{other} family of quantum states to prove this statement.
Recall that in our hypothetical $\QMA(2)$ protocol for $\NEXP$,
we are given as input a succinctly encoded $\mathsf{3Sat}$ formula $\varphi$ on $N = 2^{n}$ variables.
For each assignment $x = (x_1, \ldots, x_N)$ to the variables,
we would like there to be a corresponding proof state $\ket{\pi_x}$ on $\poly(n)$ qubits.
We want these proof states to have small overlap with each other, say $|\braket{\pi_x}{\pi_y}|^2 \leq 1/2$ for all $x \neq y$, so that two proof states $\ket{\pi_{x}}$ and $\ket{\pi_y}$ with $x \neq y$ can induce different behaviors from the verifier.
Finally, given $\poly(n, 1/\epsilon)$ copies of a state $\ket{\pi}$, we want it to be testable whether (YES) $\ket{\pi} = \ket{\pi_x}$ for some $x$ or (NO) $\ket{\pi}$ is $\epsilon$-far from all the $\ket{\pi_x}$'s.
Note that this family of states must be extremely dense, as it is fitting $2^{2^{n}}$ pairwise far-apart states into just $\poly(n)$ qubits, and as a result, we call it a \emph{dense family of testable states}. We conjecture that no such family exists.

\begin{conjecture}[No dense family of testable states]
    There is no dense family of testable states.
\end{conjecture}

As we have previously discussed,
stabilizer states,
while very similar to degree-2 phase states,
are testable with $O(1)$ copies~\cite{GNW21},
which is much smaller than the $\Omega(m)$ lower bound for phase states that we prove.
A stabilizer state is a state of the form $C \cdot \ket{0^n}$,
where $C$ is an element of the Clifford group.
We believe that a possibly interesting test case for our conjecture is states of the form $C \cdot \ket{0^n}$, where $C$ is drawn from higher levels of the Clifford hierarchy.

\section{Preliminaries}

Throughout the paper, we let $n$ denote both the dimension of the Hilbert space and the block length of the code.
We use $\log$ to denote base-$2$ logarithms and $\ln$ to denote the natural logarithm. For an event $E$, we use the convention $\mathbf{1}\{E\}$ to denote the indicator of the event, i.e., $\mathbf{1}\{E\}$ equals $1$ if and only if the event $E$ occurs. 
For $u \in \F_q^n$ and $S \subseteq \F_q^n$, let
\[
\delta(u,S) \coloneqq \min_{w \in S} \frac{d_H(u,w)}{n}
\]
denote the relative Hamming distance from $u$ to $S$. For any set $C \subseteq \F_q^n$, we define its rate to be $R(C) \coloneqq \log_q{\abs{C}}/n$. Moreover, we define its relative distance to be $\delta(C) \coloneqq \min_{x,y \in C:  x \neq y}\{\delta(x,y)\}$.

We use the convention of denoting random variables with bold notation, for example $\mathbf{a},\mathbf{b}, \mathbf{c}$. For a discrete random variable $\mathbf{X}$ with distribution $\{p(x)\}_x$, we denote the Shannon entropy of $\mathbf{X}$ as $H(\mathbf{X}) = -\sum_x p(x) \log p(x)$. Similarly for $p\in [0,1]$, we define the binary entropy function as
$h(p) = -p \log p - (1-p)\log(1-p)$. Finally, for a quantum state $\rho$, we denote its von Neumann entropy as $S(\rho) = - \tr(\rho \log \rho)$, where $\log(\cdot)$ is taken with respect to the eigenvalues of the state $\rho$. For a matrix $M$, we define $\|M\|_1$ as its trace norm, i.e., sum of its absolute eigenvalues. Throughout, we let $U_{n,q}$ denote the uniform distribution over $\F_q^n$. When unambiguous, we drop the subscript $q$ and let $U_n$ denote the distribution $U_{n,q}$.

\begin{definition}[$\eps$-tester]\label{def:tester}
Let $\calP \subseteq \C^d$ be a set of pure states and let $\dist(\cdot, \cdot)$ be a distance metric.
A state $\ket{\psi}$ is $\eps$-far from $\calP$ if $\min_{\ket\phi \in \calP} \dist(\ket\psi, \ket\phi) \geq \eps$.
An $\eps$-tester for $\calP$ using $t$ copies is a two-outcome measurement $\{M, I -M\}$ acting on $(\C^d)^{\otimes t}$ such that for every pure state $\ket{\psi}$:
\begin{itemize}
    \item If $\ket{\psi} \in \calP$ then $\tr(M \ketbra{\psi}{\psi}^{\otimes t}) \geq \frac{2}{3}$.
    \item If $\ket{\psi}$ is $\eps$-far from being in $\calP$ then $\tr(M \ketbra{\psi}{\psi}^{\otimes t}) \leq \frac{1}{3}$. 
\end{itemize}
\end{definition}

Unless otherwise stated, our distance metric will be the infidelity between quantum states.

Next, we define some notation and facts for a classical probability distribution $\calD$ over a finite set $X$. For any $x \in X$, let $\calD(x) \coloneqq \Pr_{y \sim \calD}[y=x]$.

\begin{definition}
\label{def:vu}
For any probability distribution $\calD$ on a finite set $X$, define
\[
\ket{\sqrt{\calD}} \coloneqq \sum_{x \in X}\sqrt{\calD(x)}\,\ket{x} \qquad\text{and}\qquad \diag(\calD) \coloneqq \sum_{x \in X}\calD(x)\ketbra{x}{x}.
\]
We then define
\[
\nu(\calD) \coloneqq \frac{1}{2}\norm{\ketbra{\sqrt{\calD}}{\sqrt{\calD}}-\diag(\calD)}_1.
\]
\end{definition}

\begin{lemma}
\label{lemma:nu-formula}
For any distribution $\calD$ over a finite set $X$, we have the inequality
\begin{equation*}
    \nu(\calD) \le \sqrt{1 - \sum_{x \in X}{\calD(x)^2}}
\end{equation*}
\end{lemma}
\begin{proof}
Let
\[
A(\calD)\coloneqq \ketbra{\sqrt{\calD}}{\sqrt{\calD}}-\diag(\calD).
\]
Then $A(\calD)$ has at most one positive eigenvalue. Denote this eigenvalue by $\lambda(\calD)$. Then we see that $\nu(\calD) = \lambda(\calD)$. Moreover, we see that
\begin{align*}
    \lambda(\calD)^2 &\le \tr(A(\calD)^2) \\
    &= \tr((\ketbra{\sqrt{\calD}}{\sqrt{\calD}})^2) + \tr(\diag(\calD)^2) - 2\tr(\ketbra{\sqrt{\calD}}{\sqrt{\calD}}\diag(\calD)) \\
    &= 1 + \sum_{x \in X}{\calD(x)^2} - 2\sum_{x \in X}{\calD(x)^2} \\
    &= 1 - \sum_{x \in X}{\calD(x)^2} .
\end{align*}
Taking square roots on both sides, we conclude the lemma.
\end{proof}

Finally, we will be working with the following discrete error model.

\begin{definition}[Error model]
Let $\mathcal{E}_{n,q}^t$ be a distribution of vectors $\be_t \in \F_q^n$ that is defined as
\begin{equation*}
    \be_t = \delta_{\ba_1} + \cdots + \delta_{\ba_t},
\end{equation*}
where $\ba_1, \ldots, \ba_t$ are sampled independently and uniformly at random from $[n]$.\footnote{Here $\delta_i$ is the vector with $1$ in the $i$'th entry and $0$'s elsewhere.}  When it is clear from context, we drop the subscripts and superscripts from $\mathcal{E}_{n,q}^t$ and simply denote it as $\mathcal{E}$.
\end{definition}

\subsection{Average codeword states}
\label{subsec:avg-codeword-states}

In this subsection, we define average codeword state ensembles along with some helpful facts and lemmas. Let $\ket{+^n} \coloneqq \tfrac{1}{\sqrt{n}}\sum_{i=1}^n{\ket{i}}$. For $v \in \F_q^n$, we let $Z^v \in \mathbb {C}^{n \times n}$ be the diagonal matrix with entries $(Z^v)_{i,i} = \omega^{v_i}$, where $\omega_q \coloneqq e^{2\pi i/q}$. Fix any natural number $t \in \mathbb{N}$. Consider a sample $\be_t \sim \calE^t_{n,q}$. Throughout the paper, when the values $n,q,t$ are clear from context, we will typically abbreviate $\calE^t_{n,q}$ to $\calE$, $\ket{+^n}$ to $\ket{+}$, and $\omega_q$ to $\omega$. First, we introduce the notion of a uniformly random codeword state.

\begin{definition}[Average codeword state]
    Let $V \subseteq \F_q^n$ be a linear code, and let $t$ be an integer. Then we define the mixed state
    \begin{equation*}
        \rho_V^t \coloneqq \E_{\bv \sim V}\left[\ketbra{\psi_{\bv}}{\psi_{\bv}}^{\otimes t}\right].
    \end{equation*}
    where $\ket{\psi_v} \coloneqq \frac{1}{\sqrt{n}}\sum_{i=1}^n{\omega^{v_i}\ket{i}}$. When $C$ is all of $\F_q^n$, we write $\rho^t_{\mathrm{all}} \coloneqq \rho^t_{\F_q^n}$.
\end{definition}

\noindent 
To get a better understanding of $\rho_V^t$, it will be helpful to introduce the following dephasing~channel.

\begin{definition}[Code dephasing channel]
\label{def:code-dephasing}
For any linear code $V \le \F_q^n$, define the completely dephasing channel on $n^t$-dimensional states as follows
\begin{equation*}
    \calZ_V^t(\sigma) \coloneqq \E_{\bv \sim V}\left[(Z^{\bv})^{\otimes t}\sigma (Z^{-\bv})^{\otimes t}\right],
\end{equation*}
which on input $\sigma$ conjugates it by the unitary $(Z^v)^{\otimes t}$ for a uniformly random $v \in V$.
\end{definition}

We relate these quantum channels to our average codeword state in the following fact.

\begin{fact}[Average codeword as channel output]
\label{fact:code-ensemble-as-channel-output}
    We have the identity $\rho_V^t = \calZ_V^t(\ketbra{+}{+}^{\otimes t})$.
\end{fact}

With this link at hand, we can now begin to describe the advantages of introducing this dephasing channel abstraction, the first being the following fact.

\begin{fact}[Code dephasing composition]
\label{fact:dephasing-channel-composition}
    For any two linear codes $V_1, V_2 \subseteq \F_q^n$, we have that $\calZ_{V_1}^t \circ \calZ_{V_2}^t = \calZ_{V_2}^t \circ \calZ_{V_1}^t = \calZ_{V_1 + V_2}^t$. In particular, for $\calZ_{\rm all}^t \coloneqq \calZ_{\F_q^n}^t$, we have that $\calZ_{\rm all}^t = \calZ_{\rm all}^t \circ \calZ_V^t = \calZ_V^t \circ \calZ_{\rm all}^t$.
\end{fact}

Next, we have the following handy reformulation of the dephasing channel $\calZ_V^t$ as the following projective measurement channel.

\begin{lemma}[Dephasing as projections]
\label{lemma:as-projection-channel}
    For any linear code $V \le \F_q^n$ with generator matrix $M \in \F_q^{\ell \times n}$, define
    \begin{equation*}
        T_s^V \coloneqq \{(i_1, \ldots , i_t) \in [n]^t \; : \; M(\sum_{j=1}^t{\delta_{i_j}}) = s\}.\footnote{\rm Note that the definition of $\Pi_s^V$ is contingent on the choice of a generator matrix of $V$. Since the generator matrices will always be clear from context, we will for simplicity abuse notation and write $\Pi_s^V$ in place of $\Pi_s^M$.}
    \end{equation*}
    for every $s\in \F_q^\ell$ and let $\Pi_s^V$ be the projection matrix onto the subspace $\{\ket{\psi} \in \mathbb{C}^{n^t} :\, \mathrm{supp}(\ket{\psi}) \subseteq T_s^V\}$. Then, we have the following: for every $n^t$-dimensional state $\sigma$, we have    \begin{equation*}
        \calZ_V^t(\sigma) = \sum_{s \in \F_q^\ell}{\Pi_s^V\sigma\Pi_s^V}.
    \end{equation*}
\end{lemma}
\begin{proof}
To prove this lemma, it suffices to show that for any two sequences $(a_1, \ldots , a_t), (b_1, \ldots , b_t) \in [n]^t$, we have that
\begin{equation*}
    \bra{a_1, \ldots , a_t} \calZ_V^t(\sigma) \ket{b_1, \ldots , b_t} = \sigma_{(a_1, \ldots , a_t), (b_1, \ldots , b_t)} \cdot 1\left[M(\sum_{i=1}^t{\delta_{a_i}}) = M(\sum_{i=1}^t{\delta_{b_i}})\right].
\end{equation*}
Indeed, using the identity $\E_{\bv \sim V}\left[\omega^{\langle  u,\bv \rangle}\right] = \mathbf{1}[u \in V^\perp]$ and the observation that $M$ forms the parity-check matrix of $V^\perp$, we have that
\begin{align*}
    \bra{a_1, \ldots , a_t} \calZ_V^t(\sigma) \ket{b_1, \ldots , b_t} &= \E_{\bv \sim V}\left[\bra{a_1, \ldots , a_t}(Z^v)^{\otimes t}\sigma (Z^{-v})^{\otimes t} \ket{b_1, \ldots , b_t}\right] \\
    &= \E_{\bv \sim V}\left[\omega^{\sum_{i=1}^t{v_{a_i}}} \bra{a_1, \ldots , a_t}\sigma \ket{b_1, \ldots , b_t}\omega^{-\sum_{i=1}^t{v_{b_i}}}\right] \\
    &= \sigma_{(a_1, \ldots , a_t), (b_1, \ldots , b_t)} \cdot \E_{\bv \sim V}\left[\omega^{\sum_{i=1}^t{v_{a_i}}-\sum_{i=1}^t{v_{b_i}}}\right] \\
    &= \sigma_{(a_1, \ldots , a_t), (b_1, \ldots , b_t)} \cdot \E_{\bv \sim V}\left[\omega^{\langle \sum_{i=1}^t{\delta_{a_i}}-\sum_{i=1}^t{\delta_{b_i}}, v\rangle}\right] \\
    &= \sigma_{(a_1, \ldots , a_t), (b_1, \ldots , b_t)} \cdot \mathbf{1}\left[\sum_{i=1}^t{\delta_{a_i}}-\sum_{i=1}^t{\delta_{b_i}} \in V^\perp \right] \\
    &= \sigma_{(a_1, \ldots , a_t), (b_1, \ldots , b_t)} \cdot \mathbf{1}\left[M(\sum_{i=1}^t{\delta_{a_i}}-\sum_{i=1}^t{\delta_{b_i}}) = 0 \right] \\
    &= \sigma_{(a_1, \ldots , a_t), (b_1, \ldots , b_t)} \cdot \mathbf{1}\left[M(\sum_{i=1}^t{\delta_{a_i}}) = M(\sum_{i=1}^t{\delta_{b_i}}) \right]. \qedhere
\end{align*}
\end{proof}

Now, using the projective channel measurement formulation, we give an explicit characterization of the spectrum of $\rho_V^t$, formulated in the following proposition.

\begin{proposition}[Outcome probabilities]
\label{prop:spectrum-formula}
    For any linear code $V \subseteq \F_q^n$ with generator matrix $M \in \F_q^{\ell \times n}$ and vector $s \in \F_q^\ell$, we have the identity $\norm{\Pi_s^V\ket{+}^{\otimes t}}_2^2 = \Pr[M\be_t = s]$.
\end{proposition}
\begin{proof}
    For every $s \in \F_q^\ell$, from the definition of $\Pi_s^V$ given in~\cref{lemma:as-projection-channel}, we see that
    \begin{align*}
        \norm{\Pi_s^V\ket{+}^{\otimes t}}_2^2 &= \sum_{(i_1, \ldots , i_t) \in [n]^t}n^{-t} \cdot \mathbf{1}[(i_1, \ldots , i_t) \in T_s] = \Pr[M\be_t = s]. \qedhere
    \end{align*}
\end{proof}

\section{Testing lower bounds}
\label{sec:testing-LBs}

In this section, we prove~\cref{thm:main}. We first begin in~\cref{subsec:indistinguishability} with showing that the decodability of $C^\perp$ against the error model $\calE_{n,q}^t$ implies that a uniform $C$-codeword phase state is information-theoretically indistinguishable from a uniformly random phase state. Next, in~\cref{subsec:testing-lbs}, we leverage this indistinguishability result to show testing lower bounds for such a linear code $C$. Then, in~\cref{subsec:low-deg-phase-testing-lbs}, we specialize the testing lower bound to several parameter regimes of the Reed--Muller code. Finally, we conclude in~\cref{subsec:deg-2-testing-lbs} with an $\Omega(m)$ lower bound for the $m$-variate degree-$2$ Reed--Muller code via a different approach than the one established in the preceding subsections. 

\subsection{Decodability implies indistinguishability}
\label{subsec:indistinguishability}

Let $C \subseteq \F_q^n$ be a linear code of dimension $k$ with generator matrix $G \in \F_q^{k \times n}$. In order to understand the maximum distinguishability between $\rho_C^t$ and $\rho_{\rm all}^t$, our goal in this subsection is to express the trace distance between the states $\rho_C^t$ and $\rho_{\rm all}^t$ in terms of the posterior distribution of $\be_t$ conditioned on $G\be_t$. To do so, we first recall the definition of the MAP (maximum \emph{a posteriori}) decoder.  

\begin{definition}[MAP successful decoding probability]
\label{def:map}
Let $V \subseteq \F_q^n$ be a linear code with parity-check matrix $P \in \F_q^{(n-\ell) \times n}$. Let $p_{\rm decode}^t(V)$ denote the success probability of the MAP decoder for the linear code $V$ against errors $\be_t \sim \calE_{n,q}^t$, which is the maximum achievable decoding success probability. Note that $p_{\rm decode}^t(V)$ admits the closed formula
\begin{equation*}
    p_{\rm decode}^t(V) = \sum_{s \in \F_q^{n-\ell}} \max_{x \in \F_q^n : Px = s} \{\calE_{n,q}^t(x)\}.
\end{equation*}
\end{definition}

Now, we can state the main result of this subsection.

\begin{theorem}[Decodability implies indistinguishability]
\label{thm:decodability-implies-indistinguishability}
    For every linear code $C \subseteq \F_q^n$, we have the~inequality
    \begin{equation*}
        \frac{1}{2}\norm{\rho_C^t-\rho_{\mathrm{all}}^t}_1 \le \sqrt{1 - p_{\rm decode}^t(C^\perp)^2}.
    \end{equation*}
\end{theorem}

To prove this theorem, we introduce some definitions. Let $\calS$ be the distribution of $G\be_t$ over $\F_q^k$. For each $s \in \F_q^k$, define the set of vectors $X_s \coloneqq \{x \in \F_q^n : Gx = s\}$. Moreover, let $\calE_s$ be the distribution of $\be_t \sim \calE$ when conditioned on the event $G\be_t = s$. Note that $\supp(\calE_s) \subseteq X_s$. We first begin with the following lemma.

\begin{lemma}
\label{lemma:exact-trace-distance-vu}
We have the equality
\[
\frac{1}{2}\norm{\rho_C^t-\rho_{\mathrm{all}}^t}_1 = \E_{s \sim \calS}[\nu(\calE_s)],
\]
where $\nu(\cdot)$ is defined  in~\cref{def:vu} as $\nu(\calE_s) \coloneqq \frac{1}{2}\norm{\ketbra{\sqrt{\calE_s}}{\sqrt{\calE_s}}-\diag(\calE_s)}_1$.
\end{lemma}
\begin{proof}
    Following the definition of $\Pi_s^V$ given in~\cref{lemma:as-projection-channel}, for any $s \in \F_q^k$, observe that
    \begin{equation*}
        \Pi_s^C = \sum_{x \in X_s}{\Pi_x^{\rm all}} .
    \end{equation*}
    Thus, using~\cref{fact:code-ensemble-as-channel-output} and the pairwise orthogonality of the $\Pi_s^V$'s from~\cref{lemma:as-projection-channel}, we see that
    \begin{align*}
        \frac{1}{2}\norm{\rho_C^t-\rho_{\mathrm{all}}^t}_1 &= \frac{1}{2}\NORM{\calZ_C^t(\ketbra{+}{+}^{\otimes t})-\calZ_{\rm all}^t(\ketbra{+}{+}^{\otimes t})}_1 \\
        &= \frac{1}{2}\NORM{\sum_{s \in \F_q^k}{\Pi_s^C\ketbra{+}{+}^{\otimes t}\Pi_s^C} - \sum_{x \in \F_q^n}{\Pi_x^{\rm all}\ketbra{+}{+}^{\otimes t}}\Pi_x^{\rm all}}_1 \\
        &= \frac{1}{2}\NORM{\sum_{s \in \F_q^k}{\left(\Pi_s^C\ketbra{+}{+}^{\otimes t}\Pi_s^C - \sum_{x \in X_s}{\Pi_x^{\rm all}\ketbra{+}{+}^{\otimes t}\Pi_x^{\rm all}}\right)}}_1 \\
        &= \sum_{s \in \F_q^k}{\frac{1}{2}\NORM{\Pi_s^C\ketbra{+}{+}^{\otimes t}\Pi_s^C - \sum_{x \in X_s}{\Pi_x^{\rm all}\ketbra{+}{+}^{\otimes t}\Pi_x^{\rm all}}}}_1 \\
        &= \sum_{s \in \F_q^k}{\calS(s) \cdot \frac{1}{2}\NORM{\frac{1}{\calS(s)}\Pi_s^C\ketbra{+}{+}^{\otimes t}\Pi_s^C - \frac{1}{\calS(s)}\sum_{x \in X_s}{\Pi_x^{\rm all}\ketbra{+}{+}^{\otimes t}\Pi_x^{\rm all}}}}_1 .
    \end{align*}
    Now, by~\cref{prop:spectrum-formula}, we know that $\frac{1}{\calS(s)}\Pi_s^C\ketbra{+}{+}^{\otimes t}\Pi_s^C$ forms a pure state. Moreover, by~\cref{prop:spectrum-formula}, for any $x \in X_s$, we also have that
    \begin{equation*}
        \tr\left(\frac{1}{\calS(s)} {\Pi_x^{\rm all}\ketbra{+}{+}^{\otimes t}}\Pi_x^{\rm all}\right) = \frac{1}{\calS(s)} \cdot \NORM{\Pi_x^{\rm all}\ket{+}^{\otimes t}}_2^2 = \frac{1}{\calS(s)} \cdot \Pr[\be_t = x] = \calE_s(x).
    \end{equation*}
    Since $\Pi_s^C = \sum_{x \in X_s}{\Pi_x^{\rm all}}$, we can write
    \begin{equation*}
        \frac{1}{\sqrt{\calS(s)}}\Pi_s^C\ket{+}^{\otimes t} = \sum_{x \in X_s}{\sqrt{\calE_s(x)} \cdot \frac{1}{\sqrt{\calE(x)}}\Pi_x^{\rm all}\ket{+}^{\otimes t}}.
    \end{equation*}
    Note that $\frac{1}{\sqrt{\calE(x)}}\Pi_x^{\rm all}\ket{+}^{\otimes t}$ is a pure state by~\cref{prop:spectrum-formula}. Therefore, by applying~\cref{def:vu} for $\calD = \calE_s$, we conclude that 
    \begin{equation*}
        \frac{1}{2}\NORM{\frac{1}{\calS(s)}\Pi_s^C\ketbra{+}{+}^{\otimes t}\Pi_s^C - \frac{1}{\calS(s)}\sum_{x \in X_s}{\Pi_x^{\rm all}\ketbra{+}{+}^{\otimes t}\Pi_x^{\rm all}}}_1 = \nu(\calE_s),
    \end{equation*}
    which finishes the proof.
\end{proof}

Next, we prove that being able to correct errors in the dual code causes $\E_{s \sim \calS}[\nu(\calE_s)]$ to be small.

\begin{lemma}
\label{lemma:dual-decodability-ub}
For any linear code $C \subseteq \F_q^n$, we have 
$\E_{s \sim \calS}[\nu(\calE_s)] \leq \sqrt{1 - p_{\rm decode}^t(C^\perp)^2}$.
\end{lemma}
\begin{proof}
By \cref{lemma:nu-formula}, Jensen's inequality, and linearity of expectation, we have that 
\begin{align*}
   \E_{s \sim \calS}[\nu(\calE_s)] 
   \le \E_{s \sim \calS}\left[\sqrt{1 - \sum_{x \in \F_q^n} \calE_s(x)^2}\right]
   \leq \sqrt{1 - \sum_{x \in \F_q^n} \E_{s \sim \calS}[\calE_s(x)^2]}. 
\end{align*}
Now, because $G$ is the generator matrix of $C$, it is additionally the parity-check matrix of $C^\perp$. Therefore, using the formula for $p_{\rm decode}^t(C^\perp)$ given in~\cref{def:map}, we conclude that 
\begin{align*}
\sum_{x \in \F_q^n} \E_{s \sim \calS}[\calE_s(x)^2] 
&= \sum_{x \in \F_q^n} \sum_{s \in \F_q^k} \calS(s) \calE_s(x)^2 \\
&\ge \sum_{s \in \F_q^k} \calS(s) (\max_{x \in \F_q^n}\{\calE_{s}(x)\})^2  \\
&\ge \left(\sum_{s \in \F_q^k} \calS(s) \max_{x \in \F_q^n} \calE_{s}(x)\right)^2  \\
&= p_{\rm decode}^t(C^\perp)^2 . \qedhere
\end{align*}
\end{proof}

\noindent Now, by directly combining~\cref{lemma:exact-trace-distance-vu} and~\cref{lemma:dual-decodability-ub}, we  get $  \frac{1}{4}\norm{\rho_C^t-\rho_{\mathrm{all}}^t}_1^2 \le {1 - p_{\rm decode}^t(C^\perp)^2}$, immediately concluding~\cref{thm:decodability-implies-indistinguishability}.

\subsection{Testing lower bounds from indistinguishability}
\label{subsec:testing-lbs}

With~\cref{thm:decodability-implies-indistinguishability} now established, in this subsection, we will present our testing lower bounds for $t$ copies of a codeword state whenever $C^\perp$ is decodable against $t$ random errors. Now, for $v \in \F_q^n$, recall the definition $\ket{\psi_v} \coloneqq Z^v \ket{+}$ from~\cref{subsec:avg-codeword-states}. Every linear code $C$ induces the following property of quantum states 
\[
\calP_C \coloneqq \{\ket{\psi_c} : c \in C\}.
\]
With these definitions at hand, we can now state the main theorem of this subsection.

\begin{theorem}
\label{thm:tester-LB}
Let $C \subseteq \F_q^n$ be a linear code, and let $\gamma \coloneqq \sqrt{1 - p^t_{\rm dec}(C^\perp)^2} + \frac{t(t-1)}{n}$.
Fix $\eps<1-\frac{\pi}{4}$ and $\delta>0$. For all sufficiently large $n$, if
\[
\gamma  \le \frac{1}{3}-\delta,
\]
then every $\eps$-tester for $\mathcal P_C$ requires at least $t+1$ copies.
\end{theorem}

To prove this theorem, we first show that random $C$-codeword phase states form approximate state $t$-designs. 
We then show that this state design property, along with the fact that the neighborhood around $C$-codeword phase states has negligible Haar measure,  implies the property testing lower bound. 
We establish some notation and basic facts to formalize this.
Let $\mathrm{Haar}$ denote the Haar measure on the unit sphere $\C^n$.
Define $\rho_{\rm Haar}^t \coloneqq  \mathbb{E}_{\ket{\bm{\vartheta}} \sim \mathrm{Haar}} \left[\, \ket{\bm{\vartheta}}\!\bra{\bm{\vartheta}}^{\otimes t} \,\right]$. 

\begin{definition}[$\eps$-approximate state $t$-design]
    Let $\eps \in [0,1]$, $t \in \N$, and let $\calQ$ be a distribution over $n$-qudit states.
    Let $\rho_\calQ^t \coloneqq \E_{\ket\psi \sim \calQ}[\ketbra{\psi}{\psi}^{\otimes t}]$.
    We say that $\calQ$ is an $\eps$-approximate state $t$-design if 
    \[
    \mathrm{D}_{\rm tr}(\rho_\calQ^t, \rho_{\rm Haar}^t) \le \eps.
    \]
\end{definition}

The next theorem shows that a random $\F_q$ phase state is close to $\rho_{\rm Haar}^t$. 
The case $q=2$ was established by Brakerski and Shmueli~\cite{brakerski10.1007/978-3-030-36030-6_10}; see also \cite{ananth2022pseudorandom}. We extend this result to an arbitrary prime $q$. For completeness, we provide the proof in \cref{appendix:random-fq}.

\begin{restatable}{theorem}{randomfqphase}
\label{thm:random-fq-phase}
We have the inequality 
\[
\mathrm{D}_{\rm tr}(\rho_{\rm all}^t, \rho_{\rm Haar}^t) \le \frac{t(t-1)}{n}.
\]
\end{restatable}

We can combine \cref{thm:decodability-implies-indistinguishability,thm:random-fq-phase} to show that any average $C$-codeword state is close to Haar-random, as long as $p_{\rm decode}^t(C^\perp)$ is large.

\begin{corollary}
\label{lemma:general-c-state-design}
We have the inequality
\[
\mathrm{D}_{\mathrm{tr}}(\rho_{C}^t, \rho_{\rm Haar}^t) \leq \sqrt{1 - p_{\rm decode}^t(C^\perp)^2} +  \frac{t(t-1)}{n} \eqqcolon \gamma.
\]
Therefore, the uniform distribution over $C$-codeword states is a $\gamma$-approximate $t$-design.
\end{corollary}
\begin{proof}
    Combine \cref{thm:decodability-implies-indistinguishability,thm:random-fq-phase} and apply the triangle inequality.
\end{proof}

Although our focus is on property testing lower bounds, \cref{lemma:general-c-state-design} may be of independent interest. In \cref{cor:low-degree-design}, we show that low-degree phase states yield highly accurate designs for moments much larger than the degree. We leave a broader investigation of which codes give rise to high-quality designs, as well as potential applications to quantum pseudorandomness and cryptography, to future work.

\cref{lemma:general-c-state-design} alone is insufficient to obtain property testing lower bounds. Indeed, the set of all pure states forms an exact design, yet membership in this set is trivial to test. 
To obtain lower bounds, we need an additional property: the $\eps$-neighborhood of $\calP_C$ must also have negligible Haar measure.
For a property $\calP$, we quantify this with $\eta_\eps(\calP) \coloneqq \Pr_{\ket{\psi}\sim \mathrm{Haar}(\C^n)}[\max_{\ket\phi \in \calP} |\braket{\psi}{\phi}|^2 \ge 1 - \eps]$.

\begin{lemma}
\label{lemma:tester-lb}
Suppose 
$\mathrm{D}_{\mathrm{tr}}(\rho_{C}^t, \rho_{\rm Haar}^t) \leq \gamma$. 
If 
\[
\eta_\eps(\calP_C) < \frac{1}{2} - \frac{3\gamma}{2}, 
\]
then every $\eps$-tester for $\calP_C$ requires at least $t+1$ copies.
\end{lemma}
\begin{proof}
Assume toward a contradiction that there is an $\eps$-tester using $k$ copies with $k \le t$.
Let $M$ be the accepting POVM element, so $0 \leq M \le I$ and the tester is $\{M, I-M\}$.
By definition of an $\eps$-tester, if $\ket\psi \in \calP_C$, then $\tr(M \ketbra{\psi}{\psi}^{\otimes k}) \geq 2/3$, and if $\ket\psi$ is $\eps$-far from $\calP$, then $\tr(M \ketbra{\psi}{\psi}^{\otimes k}) \le 1/3$.
Thus, by linearity, we have   
\[
\abs{\tr(M \rho_C^k) - \tr(M \rho_{\rm Haar}^k)} \leq 
\mathrm{D}_{\rm tr}(\rho_C^k, \rho_{\rm Haar}^k) \leq \gamma. 
\]
The last line uses that $\mathrm{D}_{\mathrm{tr}}(\rho_{\rm C}^k, \rho_{\rm Haar}^k) \leq \gamma$ for all $k \le t$, because partial trace is a CPTP map and the trace distance is contractive. 
Therefore, 
\begin{equation}
\label{haar-lb}
\tr(M \rho_{\rm Haar}^k) \geq  \tr(M \rho_C^k) - \gamma \geq \frac{2}{3} - \gamma.
\end{equation}

Now let $\calB_\eps \coloneqq \{ \ket{\psi} : \exists \ket\phi \in \calP_C : \abs{\braket{\psi}{\phi}}^2 < 1-\eps \}$. By definition, $\Pr_{\ket\psi \sim \mathrm{Haar}(\C^n)}[\ket\psi \in \calB_\eps] = \eta_\eps(\calP_C)$.
If $\ket\psi \not\in \calB_\eps$, then $\ket\psi$ is $\eps$-far from $\calP$. Thus, the soundness of the tester implies that the test accepts with probability at most $1/3$.
Therefore, 
\[
\tr(M \ketbra{\psi}{\psi}^{\otimes k}) \leq \mathbf{1}\{\ket\psi \in \calB_\eps\} + (1 - \mathbf{1}\{\ket\psi \in \calB_\eps\}) \cdot \frac{1}{3}.
\]
Averaging over the Haar distribution gives, 
\begin{equation}
\label{eq:up}
\E_{\ket\psi \sim \mathrm{Haar}(\C^n)}\tr(M \ketbra{\psi}{\psi}) = \tr(M \rho_{\rm Haar}^k) \leq \eta_\eps(\calP_C) + (1 - \eta_\eps(\calP_C)) \cdot \frac{1}{3} \leq \frac{1}{3} + \frac{2}{3}\eta_\eps(\calP_C).
\end{equation}

Finally, combining \cref{haar-lb,eq:up} yields 
\[
\frac{2}{3} - \gamma \leq \frac{1}{3} + \frac{2}{3}\eta_\eps(\calP_C) 
\iff 
\frac{1 -3 \gamma}{2} \leq \eta_\eps(\calP_C). 
\]
Thus any $\eps$-tester using $k \le t$ copies forces 
$\frac{1 -3 \gamma}{2} \leq \eta_\eps(\calP_C)$. 
It follows that if  $\frac{1 -3 \gamma}{2} > \eta_\eps(\calP_C)$, such a tester cannot exist. So every $\eps$-tester must satisfy $k \geq t + 1$.
\end{proof}

With \cref{lemma:tester-lb} established, it remains to bound the quantity $\eta_\eps(\calP_C)$.
A standard approach is to apply a union bound. 
Indeed, for any fixed state $\ket{\varphi}$, 
$\Pr_{\ket{\psi} \sim \mathrm{Haar}}[\abs{\braket{\psi}{\varphi}}^2 \ge 1 - \eps] = \eps^{2(n-1)}$. 
It follows that $\eta_\eps(\calP_C) \le \abs{C} \eps^{2(n-1)}$.
Instead, we prove a bound that avoids any dependence on $\abs{C}$. This is particularly useful because it will yield lower bounds that are independent of the dimension of $C$. To establish it, we will need the following series of results. 

\begin{fact}
\label{fact:beta-dist}
Let $\ket{\psi} = \sum_{i = 1}^n \bm{\alpha}_i \ket{i}$ be a Haar-random pure state. For every $i \in [n]$, \[\abs{\bm{\alpha}_i}^2 \sim \mathrm{Beta}(1, n-1).\]    
\end{fact}
\begin{proof}
A Haar-random state in \(\mathbb C^n\) can be generated by normalizing a standard complex Gaussian vector:
\[
\ket{\psi}
=
\frac{(\bz_1,\dots,\bz_n)}{\sqrt{\sum_{j=1}^n |\bz_j|^2}},
\]
where \(\bz_1,\dots,\bz_n\) are iid \(\mathcal{CN}(0,1)\). Therefore
\[
|\bm{\alpha}_i|^2
= \frac{|\bz_i|^2}{\sum_{j=1}^n |\bz_j|^2}.
\]
Since \(|\bz_1|^2,\dots,|\bz_n|^2\) are iid $\mathrm{Exp}(1)=\Gamma(1,1)$, the vector
\[
\left(
\frac{|\bz_1|^2}{\sum_{j=1}^n |\bz_j|^2},
\dots,
\frac{|\bz_n|^2}{\sum_{j=1}^n |\bz_j|^2}
\right)
\]
has the Dirichlet distribution $\mathrm{Dirichlet}(1,\dots,1)$. Hence each marginal is $|\bm{\alpha}_i|^2 \sim \mathrm{Beta}(1,n-1)$, as claimed.
\end{proof}

\begin{lemma}[Wendel-Gautschi Inequalities {\cite{wendel1948note,gautschi1959some}}]
\label{lemma:wendel}
For $x > 0$ and $0 < s < 1$, 
\[
\left(\frac{x}{x+s} \right)^{1-s} \le \frac{\Gamma(x+s)}{x^s \Gamma(x)} \le 1, 
\]
where, for $z \in \C$ with $\Re(z) > 0$,
$\Gamma(z) = \int_0^\infty t^{z-1} e^{-t} \, dt$ is the gamma function.
\end{lemma}

\begin{lemma}[L\'{e}vy's Lemma, see e.g. \cite{Gerken13measureconcentration}]
\label{lem:levy}
Let $\mathbb{S}^n$ denote the set of all $n$-dimensional pure quantum states, and let $f: \mathbb{S}^n \to \mathbb{R}$ be $L$-Lipschitz, meaning that $\abs{f(\ket{\psi}) - f(\ket{\varphi})} \le L \cdot \norm{\ket{\psi} - \ket{\varphi}}_2$. Then:
\[
\Pr_{\ket{\psi} \sim \rm{Haar}} \left[\abs{f(\ket{\psi}) - \E[f]} \ge \eps  \right] \le 2\exp\left(-\frac{n \eps^2}{9\pi^3 L^2} \right).
\]
\end{lemma}

We now prove our the main technical lemma that will imply a bound on $\eta_\eps(\calP_C)$.

\begin{lemma}\label{lemma:small-haar-measure}
Define $\calP_{\rm flat} = \{ \frac{1}{\sqrt n} \sum_{i = 1}^n e^{i \theta_i} \ket{i} : \theta_i \in \R\}$.
    For every fixed $\eps < 1 - \tfrac{\pi}{4}$, there exists $c_\eps > 0$ and $n_\eps$ such that for all $n \ge n_\eps$, 
    \[
    \Pr_{\ket{\psi} \sim \mathrm{Haar}} \left[\max_{\ket{\theta} \in \calP_{\rm flat}} \abs{\braket{\psi}{\theta}}^2 \ge 1 - \eps \right] \le 2 e^{-c_\eps n}.
    \]
\end{lemma}
\begin{proof}
For $\ket{\psi} = \sum_{i = 1}^n \alpha_i \ket{i}$, define $L(\ket\psi) = \tfrac{1}{\sqrt n} \sum_{i = 1}^n \abs{\alpha_i}$. 
Then
\[
\max_{\ket{\theta} \in \calP_{\rm flat}} \abs{\braket{\psi}{\theta}}^2 
= \max_{\theta_i} \Abs{ \frac{1}{\sqrt n} \sum_{i = 1}^n e^{-i\theta_i} \alpha_i}^2 = L(\ket{\psi})^2.
\]
By symmetry, we have 
\[
\E[L(\ket\psi)] = \frac{1}{\sqrt{n}} \sum_{i = 1}^n \E \abs{\alpha_i} = \sqrt{n} \E\abs{\alpha_1}.
\]
For Haar-random $\ket{\psi}$, we have by \cref{fact:beta-dist} that $\abs{\bm{\alpha}_1}^2 \sim \mathrm{Beta}(1, n-1)$, which has density function $f(x) = (n-1)(1-x)^{n-2}$ for $x \in [0,1]$.
Therefore 
\[
\E\abs{\bm{\alpha}_1} = (n-1) \int_0^1 x^{1/2}(1-x)^{n-2} dx =  \frac{\Gamma(3/2)\Gamma(n)}{\Gamma(n + 1/2)}, 
\]
where $\Gamma(\cdot)$ is the gamma function.
Using the fact that $\Gamma(3/2) = \sqrt{\pi}/2$ and \cref{lemma:wendel} with $x = n$ and $s = 1/2$, we have 
\[
\frac{\sqrt{\pi}}{2} \le \E L(\ket{\psi}) \le \frac{\sqrt\pi}{2} \sqrt{1 + \frac{1}{2n}}. 
\]
Now set $a_\eps \coloneqq \sqrt{1 - \eps}$. Since $\eps < 1 - \pi/4$, we have $a_\eps > \sqrt{\pi}/2$. Define $\Delta_\eps \coloneqq a_\eps - \sqrt{\pi}/2 > 0$.
By the upper bound on $\E L(\ket\psi)$, there exists an $n_\eps$ such that for all $n \ge n_\eps$, 
\[
\E [L(\ket\psi)] \le \frac{\sqrt{\pi}}{2} \sqrt{1 + \frac{1}{2n}} \le \frac{\sqrt{\pi}}{2} + \frac{\Delta_\eps}{2} = a_\eps - \frac{\Delta_\eps}{2}. 
\]
We next note that $L$ is $1$-Lipschitz on the unit sphere. Indeed, for $\ket{\psi} = \sum_i \alpha_i \ket{i}$ and $\ket{\varphi} = \sum_i \beta_i \ket{i}$, we~have 
\begin{align*}
\abs{L(\ket\psi)-L(\ket\varphi)}
&= \frac1{\sqrt n} \left| \sum_{i=1}^n |\alpha_i|-\sum_{i=1}^n |\beta_i| \right| \\
&\le \frac1{\sqrt n}\sum_{i=1}^n \bigl||\alpha_i|-|\beta_i|\bigr| \\
&\le \frac1{\sqrt n}\sum_{i=1}^n |\alpha_i-\beta_i| \\
&\le \left(\sum_{i=1}^n |\alpha_i-\beta_i|^2\right)^{1/2} \\
&= \norm{\ket\psi-\ket\varphi}_2.
\end{align*}

Then
\begin{align*}
\Pr_{\ket{\psi} \sim \mathrm{Haar}} \left[\max_{\ket{\theta} \in \calP_{\rm flat}} \abs{\braket{\psi}{\theta}}^2 \ge 1 - \eps \right]
&= \Pr_{\ket\psi \sim \mathrm{Haar}}[L(\ket{\psi})^2 \geq 1 - \eps] \\
&= \Pr_{\ket\psi \sim \mathrm{Haar}}[L(\ket{\psi}) \geq \sqrt{1 - \eps}] \\
&= \Pr_{\ket\psi \sim \mathrm{Haar}}[L(\ket{\psi}) \geq a_\eps] \\
&\le \Pr_{\ket\psi \sim \mathrm{Haar}}[\abs{L(\ket{\psi}) - \E L(\ket{\psi})} \geq \frac{\Delta_\eps}{2}]  \\
&\le 2 \exp\left(- \frac{n \Delta_\eps^2}{36 \pi^3} \right).
\end{align*}
The second-to-last line follows because, for all $n \ge n_\eps$, $\E L(\ket{\psi}) \le a_\eps - \tfrac{\Delta_\eps}{2}$. Thus the event $L(\ket\psi) \ge a_\eps$ implies $L(\ket\psi) - \E L(\ket\psi) \ge \Delta_\eps / 2$.
The last line follows from L\'evy's Lemma (\cref{lem:levy}).
Therefore the claim holds with $c_\eps = \tfrac{\Delta_\eps^2}{36\pi^3}$ where $\Delta_\eps = \sqrt{1 - \eps} - \tfrac{\sqrt{\pi}}{2}$.
\end{proof}

With that, we can now establish the main theorem of this subsection.

\begin{proof}[Proof of \cref{thm:tester-LB}]
Because $\calP_C \subseteq \calP_{\rm flat}$, we have $\eta_\eps(\calP_C) \le \eta_\eps(\calP_{\rm flat})$.
Thus, by \cref{lemma:small-haar-measure}, there exist constants
$c_\eps>0$ and $n_\eps$ such that for all $n\ge n_\eps$,
$\eta_\eps(\calP_C)\le 2e^{-c_\eps n}$.

Now fix $\delta>0$. Choose $n_{\eps,\delta}\ge n_\eps$ large enough so that
for all $n\ge n_{\eps,\delta}$,
$2e^{-c_\eps n}<\frac{3\delta}{2}$.
Then, for all $n\ge n_{\eps,\delta}$,
$\eta_\eps(\calP_C) < \frac{3\delta}{2}.$ By \cref{lemma:general-c-state-design}, we have that $\mathrm{D}_{\mathrm{tr}}(\rho_C^{t},\rho_{\mathrm{Haar}}^{t})  \le \gamma.$
Thus, by \cref{lemma:tester-lb} and our bound $\eta_\eps(\calP_C) < \frac{3\delta}{2}$, every $\eps$-tester for $\calP_C$ requires $t+1$ copies if 
\[
\frac{3\delta}{2} < \frac{1}{2} - \frac{3\gamma}{2} \iff \gamma < \frac{1}{3} - \delta.
\]
This completes the proof.
\end{proof}

\subsection{Application to degree-\texorpdfstring{$d$}{d} phase states}
\label{subsec:low-deg-phase-testing-lbs}

In this subsection, we will apply the machinery developed in~\cref{subsec:indistinguishability,subsec:testing-lbs} to Reed--Muller codeword states to prove that there are no low-degree tests for quantum states.
To do so, we will use the following statement, which, when combined with \cref{thm:tester-LB}, will imply our lower bounds.
In what follows, $D_{d,q}(r)$ denotes the number of monomials in $r$ variables of total degree at most $d$, with each variable having degree at most $q-1$.

\begin{restatable}{corollary}{decodebounds}
\label{corollary:rm-decode-bounds}
Let $q$ be a prime number, $m \ge 1$, $0 \le d \le \lfloor \tfrac{m(q-1)-2}{2}\rfloor$, and $0 \le r \le m$.
Let $C = \RM_q[m,2d+1]$, let $t \le D_{d,q}(r)$, and let $\be_t \sim \calE_{q^m,q}^t$. 
Then
\[
\sqrt{1 - p_{\rm decode}^t(C^\perp)^2}
\le \sqrt{2tq^{r-m}}.
\]
\end{restatable}

The proof of~\cref{corollary:rm-decode-bounds} is deferred to the end of~\cref{sec:asw}. For now, we will leverage this corollary and begin proving property testing lower bounds for Reed--Muller codeword states.
The following lemma is the most general lower bound we can prove.

\begin{lemma}
\label{lemma:general-rm-bound}
    Fix $\eps < 1 - \pi/4$, and suppose $d_0 \le \lfloor \tfrac{m(q-1)-2}{2}\rfloor$. 
    For all sufficiently large $m$, every $\eps$-tester for $\calP_{\RM_q[m,d]}$ requires 
    \[
    \Omega \left(\max_{0 \le r \le m} \max \left\{D_{d_0,q}(r), q^{m-r}, q^{m/2} \right\} \right) 
    \]
    copies.
\end{lemma}
\begin{proof}
   Let $r \in \{0,\dots, m\}$ and let $T_r \coloneq \min\{D_{d_0,q}(r), q^{m-r}, q^{m/2}\}$. 
   It suffices to show that, for an absolute constant $0 <c < 1$, no tester using $t \coloneqq \lfloor c T_r\rfloor$ copies exists.

   Let $C \coloneqq \RM_q[m,d]$ and $C' \coloneqq \RM_q[m, 2d_0 -1]$. Since $2d_0 + 1 \le d$, we have $C' \subseteq C$, and hence $C^\perp \subseteq {C'}^\perp$.
   Therefore, decoding from $C^\perp$ is at least as easy as decoding from ${C'}^\perp$, so $p^t_{\rm decode}(C^\perp) \ge p^t_{\rm decode}({C'}^\perp)$.
   Because $t \le T_r \le D_{d_0,q}(r)$, \cref{corollary:rm-decode-bounds} implies 
   \[ 
   \sqrt{1 - p^t_{\rm decode}({C'}^\perp)^2} \le \sqrt{2t q^{r-m}}.
   \]
   Because $p^t_{\rm decode}(C^\perp) \ge p^t_{\rm decode}({C'}^\perp)$, we get the same bound on $\sqrt{1 - p^t_{\rm decode}({C}^\perp)^2}$.

    By \cref{lemma:general-c-state-design}, we have $\mathrm{D}_{\rm tr}(\rho_C^t, \rho_{\rm Haar}^t) \le \gamma$ with 
    \[
    \gamma \le \sqrt{1 - p_{\rm decode}^t(C^\perp)^2} + \frac{t (t-1)}{q^m} \le \sqrt{2 t q^{r-m}} + \frac{t^2}{q^m}.
    \]
    Our choice of $t$ guarantees $t \le c q^{m-r}$ and $t \le c q^{m/2}$.
    Substituting these bounds yields 
    \[
    \gamma \le \sqrt{2 c} + c^2. 
    \]
    Therefore, by choosing $0 < c < 1$ sufficiently small makes $\gamma < 1/3 - \delta$ for some fixed $\delta > 0$. 
   \cref{thm:tester-LB} then implies that every $\eps$-tester requires at least $t+1 = \Omega(T_r)$ copies. Maximizing over $r$ completes the proof.
\end{proof}

Though seemingly opaque, we will instantiate \cref{lemma:general-rm-bound} to derive concrete testing lower bounds.

\begin{theorem}\label{thm:rm-testing-lb}
   Fix $\eps < 1 - \pi/4$, and set $d_0 = \lfloor \tfrac{d-1}{2}\rfloor$. 
   For all sufficiently large $m$, every $\eps$-tester for $\calP_{\RM_q[m,d]}$ requires
   \[
   \Omega \left(\binom{\lfloor m/2 \rfloor}{\le d_0} \right)
   \]
   copies.
\end{theorem}
\begin{proof}
    We apply \cref{lemma:general-rm-bound} with $r = \lfloor m/2 \rfloor$.
    Since $D_{d_0,q}(r) \le q^r$ and $r = \lfloor m/2 \rfloor$, we have 
    \[
    D_{d_0, q}(r) \le q^{m-r} \qquad \text{and} \qquad D_{d_0,q}(r) \le q^{m/2}.
    \]
    Therefore, \cref{lemma:general-rm-bound} implies an $\Omega(D_{d_0,q}(r))$ testing lower bound.

    It remains to lower bound $D_{d_0,q}(r)$. We achieve this lower bound by counting every squarefree monomial of degree at most $d_0$, which clearly lower bounds $D_{d_0,q}(r)$: 
    \[
    D_{d_0,q}(r) \ge \sum_{j=0}^{\min\{r, d_0\}} \binom{r}{j} = \binom{r}{\le d_0}.\footnote{In the regime where $q > d$, one can strengthen this argument to obtain a slightly improved bound of $\Omega(\binom{\lfloor m/2 \rfloor + d_0}{d_0})$.}\qedhere
    \]
\end{proof}

In particular, in the regime in which low-degree tests are used in PCP constructions, we see that a superpolynomial number of copies are necessary.

We conclude this section by discussing state designs. Our testing lower bound implicitly establishes that random low-degree phase states form approximate state designs. We make this connection explicit in the following theorem.

\begin{theorem}
\label{thm:rm-designs}
    Let $q$ be a prime number, let $C = \RM_q[m,d]$, and set $d_0 = \lfloor \tfrac{d-1}{2}\rfloor$.
    Define 
    \[
    M \coloneqq \binom{\lfloor m/2 \rfloor}{d_0} \qquad \text{and} \qquad \gamma_\star \coloneq \sqrt{\frac{M}{q^{m/2}}}.
    \]
    For every $\gamma \in (0,1)$, the uniform ensemble of degree-$d$ phase states is a $\gamma$-approximate state $t$-design in trace distance for 
    \[
    t \le c\cdot \begin{cases}
        M, & \gamma\ge \gamma_\star*,
    \\[1mm]
        \gamma^2 q^{m/2}, & \gamma<\gamma_\star
    \end{cases},
    \]
    where $c > 0$ is a universal constant.
\end{theorem}
\begin{proof}
    We apply \cref{lemma:general-c-state-design} and then \cref{corollary:rm-decode-bounds} with $r = \lfloor m/2 \rfloor$. For every $t \le D_{d_0, q}(r)$, we get
    \[
    \mathrm{D}_{\rm tr}(\rho^t_C, \rho^t_{\rm Haar}) \le \sqrt{2 t q^{r-m}} + \frac{t^2}{q^m} \le O\left(\sqrt{\frac{t}{q^{m/2}}}\right) + \frac{t^2}{q^m}.
    \]
    
    As argued in \cref{thm:rm-testing-lb}, we have $D_{d_0,q}(r) \ge M$.
    Therefore, we may take $t$ up to a constant multiple of $M$, subject only to making sure the two error terms are at most $O(\gamma)$.
    Observe that if $t \le c \gamma^2 q^{m/2}$, then 
    \[
    \sqrt{\frac{t}{q^{m/2}}} \le O(\sqrt{c}) \gamma \qquad \text{and} \qquad \frac{t^2}{q^m} \le c^2 \gamma.
    \]
    By choosing $c > 0$ sufficiently small gives $\mathrm{D}_{\rm tr}(\rho^t_C, \rho^t_{\rm Haar}) \le \gamma$.
    Thus we have a $\gamma$-approximate $t$-design with 
    \[
    t \le O\left(\min\{M, \gamma^2 q^{m/2} \} \right). 
    \]
    The two cases claimed in the result follow by noting the equality $M = \gamma^2 q^{m/2}$ occurs at $\gamma = \gamma_\star$.
\end{proof}

The usefulness of \cref{thm:rm-designs} depends on the application, and we leave a systematic exploration of such applications to future work. We do note, however, that exponentially accurate designs can be obtained from low-degree phase states, as the following corollary shows.

\begin{corollary}
\label{cor:low-degree-design}
   In the setting of \cref{thm:rm-designs}, suppose further that $M = 2^{o(m)}$. Then for any $a < \frac{1}{4} \log_2 q$, the uniform ensemble of degree-$d$ phase states is a $2^{-am}$-approximate state $t$-design for $t = \Omega(M)$.
\end{corollary}
\begin{proof}
    Set $\gamma = 2^{-am}$ in \cref{thm:rm-designs}. We get that 
    \[
    \gamma^2 q^{m/2} = 2^{(\frac{1}{2}\log_2 q - 2a)m}.
    \]
    Because $a < \frac{1}{4} \log_2 q$, we have 
    \[
    \frac{1}{2} \log_2 q - 2a > 0.
    \]
    Thus $\gamma^2 q^{m/2} = 2^{\Omega(m)}$.
    On the other hand, we have by supposition that $M = 2^{o(m)}$. Therefore, for sufficiently large $m$, $M \le \gamma^2 q^{m/2}$. Thus, we are in the case of \cref{thm:rm-designs} where the ensemble is a $2^{-am}$-approximate $t$-design for $t = \Omega(M)$.
\end{proof}

Of particular note is the sharp transition between degrees $2$ and $3$. Random degree-$2$ phase states form a $3$-design, whereas random degree-$3$ phase states already form an $n$-design.

\subsection{Degree-2 phase states}
\label{subsec:deg-2-testing-lbs}
Interestingly, our lower bounds in the previous section become trivial in the case $d=2$. Since degree-$2$ phase states are a subset of stabilizer states, which are known to be testable using $6$ copies~\cite{GNW21}, it is conceivable that degree-$2$ phase states are also testable using $O(1)$ copies.  Nevertheless, we show that testing degree-$2$ phase states still requires $\Omega(m)$ copies, matching the $O(m)$ upper bound to learn the entire state~\cite{Rot09}.
Our lower bound follows from a different argument. 
Instead of comparing the ensemble $\rho_C$, where $C = \RM_q[m,2]$, with the all-phase state ensemble $\rho_{\rm all}$, we obtain a new ensemble by truncating the support of phase states to a randomly chosen hyperplane.

We adopt notation that is more convenient for the degree-2 case. Let $1_t \in \F_q^t$ denote the all-ones vector over $\F_q$ of length $t$.
Given $x_1, \dots, x_t \in \F_q^m$, we let $X \in \F_q^{m \times t}$ be the matrix whose $i$'th column is $x_i$.
Define $\ket{X} \coloneqq \ket{x_1, \dots, x_t}$. For each $Q \in \F_q^{m \times m}$ and $\ell \in \F_q^m$, define the quadratic polynomial $f_{Q,\ell}: \F_q^m \to \F_q$ as
\[
f_{Q,\ell}(x) \coloneqq x^\top Q x + \langle \ell, x\rangle
\]
and let $\ket{\psi_{f_{Q,\ell}}} \coloneqq q^{-m/2}\sum_{x \in \F_q^m}{\omega^{f_{Q,\ell}(x)}\ket{x}}$.
Note that this covers all $m$-variate degree-$2$ polynomials over $\F_q$ with no constant term, which we ignore because it only changes the corresponding degree-$2$ phase state by a global phase. Moreover, for uniformly random $Q \in \F_q^{m \times m}$ and $\ell \in \F_q^m$, $f_{Q,\ell}$ will be a uniformly random degree-$2$ polynomial with no constant term.
Define $\rho_0 \coloneqq \E_{\boldsymbol{f}}[\ketbra{\psi_{\boldsymbol{f}}}{\psi_{\boldsymbol{f}}}^{\otimes t}]$, which is the average over all degree-$2$ phase states.

\begin{lemma}
\label{lemma:deg-2-yes-decomp}
The state $\rho_0$ can be written as
\[
\rho_0 = q^{-mt}\sum_{X,Y\in \F_q^{m \times t}} \mathbf{1}\{XX^\top = YY^\top \text{and }  X1_t = Y1_t\}\ketbra{X}{Y}.
\]
\end{lemma}
\begin{proof}
For uniformly random $\bQ \sim \F_q^{m \times m}$ and $\boldsymbol{\ell} \sim \F_q^m$, using the identity $\E_{\bv \sim \F_q^a}\left[\omega^{\langle  u,\bv \rangle}\right] = \mathbf{1}\{u = 0^a\}$, we have that 
\begin{align*}
    \rho_0 &= q^{-mt} \E_{\bQ, \boldsymbol{\ell}}\left[\sum_{X, Y \in \F_q^{m \times t}}{\omega^{\sum_{i=1}^t{f_{\bQ,\boldsymbol{\ell}}(x_i)} - \sum_{i=1}^t{f_{\bQ,\boldsymbol{\ell}}(y_i)}} \ketbra{X}{Y}}\right] \\
    &= q^{-mt} \sum_{X, Y \in \F_q^{m \times t}}{\E_{\bQ, \boldsymbol{\ell}}\left[\omega^{\sum_{i=1}^t{(\langle \bQ, x_ix_i^\top\rangle + \langle \boldsymbol{\ell}, x_i\rangle)} - \sum_{i=1}^t{(\langle \bQ, y_iy_i^\top\rangle + \langle \boldsymbol{\ell}, y_i\rangle)}}\right] \ketbra{X}{Y}} \\
    &= q^{-mt} \sum_{X, Y \in \F_q^{m \times t}}{\E_{\bQ, \boldsymbol{\ell}}\left[\omega^{\langle \bQ, XX^\top\rangle + \langle \boldsymbol{\ell}, X1_t \rangle - \langle \bQ, YY^\top\rangle - \langle \boldsymbol{\ell}, Y1_t \rangle} \right] \ketbra{X}{Y}} \\
    &= q^{-mt} \sum_{X, Y \in \F_q^{m \times t}}{\E_{\bQ}\left[\omega^{\langle XX^\top - YY^\top, \bQ\rangle}\right] \E_{\boldsymbol{\ell}}\left[[\omega^{\langle X1_t - Y1_t, \boldsymbol{\ell}\rangle}\right]\ketbra{X}{Y}} \\
    &= q^{-mt} \sum_{X, Y \in \F_q^{m \times t}}{\mathbf{1}\left\{XX^\top - YY^\top = 0\right\} \mathbf{1}\left\{X1_t - Y1_t = 0\right\}\ketbra{X}{Y}} \\
    &= q^{-mt} \sum_{X, Y \in \F_q^{m \times t}}{\mathbf{1}\left\{XX^\top = YY^\top \text{and } X1_t = Y1_t\right\}\ketbra{X}{Y}}. \qedhere
\end{align*} 
\end{proof}

We now define an ensemble of states, each of which is $(1-1/q)$-far from all degree-$2$ phase states.
For nonzero $v \in \F_q^m$, define the hyperplane $H_v \coloneqq \{x \in \F_q^m: \langle x , v \rangle = 0\}$ and the projector $P_v \coloneqq \sum_{x \in H_v}\ketbra{x}{x}$.
For an $m$-variate degree-$2$ $\F_q$-polynomial $f$, define $\ket{\psi_{v,f}} \coloneqq q^{-(m-1)/2}\sum_{x \in H_v} \omega^{f(x)} \ket{x}$. It is easy to see that 
\begin{equation}
\label{eq:proj-truncation}
\ket{\psi_{v,f}} = \sqrt{q} P_v \ket{\psi_f}.
\end{equation}
Define the ensemble $\rho_1 \coloneqq \E_{\bv \neq 0^m, \boldsymbol{f}}[\ketbra{\psi_{\bv,f}}{\psi_{\bv,f}}^{\otimes t}]$. 
\cref{eq:proj-truncation} immediately implies the following.
\begin{fact}\label{rho1-overlap}
The state $\rho_1$ can be written as
\[
\rho_1 = q^t \E_{\bv \neq 0^m}[P_{\bv}^{\otimes t} \rho_0 P_{\bv}^{\otimes t}], 
\]
and
\[
\braket{X}{\rho_1|Y} = q^t \E_{\bv \neq 0^m} \left[\mathbf{1}\{X \subseteq H_{\bv}\} \mathbf{1}\{Y \subseteq H_{\bv}\} \right] \braket{X}{\rho_0 |Y},
\]
where $X \subseteq H_v$ means that every column of $X$ is contained in the hyperplane $H_v$.
\end{fact}
\begin{proof}
By \cref{eq:proj-truncation}, $\ket{\psi_{v,f}}^{\otimes t} = q^{t/2} P_v^{\otimes t}\ket{\psi_f}^{\otimes t}$, 
and hence
\[
\ketbra{\psi_{v,f}}{\psi_{v,f}}^{\otimes t} = q^t P_v^{\otimes t}\ketbra{\psi_f}{\psi_f}^{\otimes t}P_v^{\otimes t}.
\]
Averaging first over $f$ and then over $v \neq 0^m$ gives
\[
\rho_1 = q^t \E_{\bv \neq 0^m}\left[P_{\bv}^{\otimes t}\rho_0 P_{\bv}^{\otimes t}\right].
\]
Finally, note that $P_v^{\otimes t}\ket{X} = \mathbf{1}\{X\subseteq H_v\}\ket{X}$, so 
\[
\bra{X}\rho_1\ket{Y} = q^t \E_{\bv \neq 0^m}\!\left[\mathbf{1}\{X\subseteq H_{\bv}\}\mathbf{1}\{Y\subseteq H_{\bv}\}\right]\bra{X}\rho_0\ket{Y}.
\qedhere
\]
\end{proof}
Next, we show that $\rho_0$ and $\rho_1$ admit a block diagonal structure.
Define the projectors $\Pi_{=t} \coloneqq \sum_{X \in \F_q^{m \times t} : \rank(X) = t} \ketbra{X}{X}$ and $\Pi_{<t}\coloneqq I - \Pi_{=t}$.
We need the following basic fact.

\begin{fact}\label{fact:Y-full-too}
If $X, Y \in \F_q^{m \times t}$ satisfy $\rank(X) = t$ and $XX^\top = YY^\top$, then $\rank(Y) = t$, $\col(X) = \col(Y)$.
\end{fact}
\begin{proof}
We first observe the following. If $\rank(X) = t$, then $X^T: \F_q^m \to \F_q^t$ is surjective, so $\im(XX^T) = \im(X) = \col(X)$. 
Hence,  $\rank(XX^T) = t$. Using this we also have that $\col(XX^\top) = \col(X)$. Because $XX^\top = Y Y^\top$, we have $\col(YY^\top) = \col(XX^\top) = \col(X)$. It is always true that $\col(YY^\top) \subseteq \col(Y)$. Thus, $\col(X) \subseteq \col(Y)$. Because $\dim(\col(X)) = t$ and $Y$ only has $t$ columns, we conclude that $\dim(\col(Y)) = t$. Therefore, $\rank(Y) = t$ and $\col(Y) = \col(X)$.
\end{proof}

We can now prove that the ensembles $\rho_0$ and $\rho_1$ are block diagonal with respect to the decomposition into linearly independent and rank-deficient blocks. 

\begin{lemma}
The trace distance between $\rho_0,\rho_1$ can be expressed as follows
\[
\norm{\rho_0 - \rho_1}_{1} = 
\norm{\Pi_{=t}(\rho_0 - \rho_1) \Pi_{=t}}_1
+ \norm{\Pi_{<t}(\rho_0 - \rho_1) \Pi_{<t}}_1.
\]
\end{lemma}
\begin{proof}
To this end, we first prove that
  \[
   \Pi_{=t} \rho_0 \Pi_{<t}  =
   \Pi_{<t} \rho_0 \Pi_{=t}  
   = 0 \qquad\text{and}\qquad 
   \Pi_{=t} \rho_1 \Pi_{<t} = 
   \Pi_{<t}\rho_1  \Pi_{=t} 
   = 0.
    \]
    To see this, suppose $\braket{X}{\rho_0 | Y} \neq 0$. By~\cref{lemma:deg-2-yes-decomp}, we have that \[
\braket{X}{\rho_0|Y} = q^{-mt} \cdot \mathbf{1}\{XX^T = Y Y^T \text{and } X1_t = Y1_t\},
\] which implies in particular that $XX^\top = YY^\top$.
Thus, by \cref{fact:Y-full-too}, if $\rank(X) = t$, then $Y$ also has full rank. 
Thus, it is impossible to have nonzero terms between the full-rank and rank-deficient blocks of $\rho_0$. I.e.,  $\Pi_{=t} \rho_0 \Pi_{<t}  = \Pi_{<t} \rho_0 \Pi_{=t} = 0$.
If $\braket{X}{\rho_1|Y} \neq 0$, then by \cref{rho1-overlap} we must have that $\braket{X}{\rho_0|Y} \neq 0$, so the same argument applies. Using this one can see that $\rho_0-\rho_1$ is block diagonal with respect to the decomposition $\Pi_{=t}+\Pi_{<t}$.
The trace norm of a block-diagonal operator is the sum of the trace norms of its diagonal blocks, which gives the claim.
\end{proof}    

To get a bound on the trace distance between $\rho_0$ and $\rho_1$, we now bound the $1$-norm of the two blocks separately.
\begin{lemma}
\label{lemma:deg-2-LI-block-bound}
One can bound the two trace norms in the previous lemma as
\[
\norm{\Pi_{=t}(\rho_0 - \rho_1) \Pi_{=t}}_1
\leq q^{t-m}, \qquad \norm{\Pi_{<t}(\rho_0 - \rho_1) \Pi_{<t}}_1 \leq (q+1) \cdot q^{t-m}.
\]
\end{lemma}
\begin{proof}
We first prove the first inequality. Assume $\rank(X) = \rank(Y) = t$ and $\braket{X}{\rho_0|Y} \neq 0$. By \cref{fact:Y-full-too}, $\col(X) = \col(Y)$. The requirement that $X \subseteq H_v$ is equivalent to $v^\top X = 0$. Because $\col(X) = \col(Y)$, we have $\mathbf{1}\{X \subseteq H_v\} \mathbf{1}\{Y \subseteq H_v\} = \mathbf{1}\{v^\top X = 0\}$. Since $\dim(\col(X)) = t$, there are $q^{m-t}-1$ nonzero vectors $v$ orthogonal to it. Taking the expectation over $q^m-1$ nonzero choices for $v$, 
\begin{equation*}
    \E_{\bv \neq 0^m}[ \mathbf{1}\{X \subseteq H_{\bv}\} \mathbf{1}\{Y \subseteq H_{\bv}\}] =\E_{\bv \neq 0^m}[ \mathbf{1}\{\bv^\top X = 0\}] = \frac{q^{m-t}-1}{q^m - 1}.
\end{equation*}
    We have shown that $\Pi_{=t}\rho_1 \Pi_{=t} = \frac{q^m - q^t}{q^m -1} \Pi_{=t}\rho_0\Pi_{=t}$.
    Because $\Pi_{=t}\rho_0 \Pi_{=t}$ is PSD, its trace norm equals its trace. We have 
\begin{align*}
    \norm{\Pi_{=t}(\rho_0 - \rho_1) \Pi_{=t}}_1 &= \left(1 - \frac{q^m - q^t}{q^m - 1}\right) \norm{\Pi_{=t}\rho_0 \Pi_{=t}}_1 \\
    &= \left(1 - \frac{q^m - q^t}{q^m - 1}\right) \tr(\Pi_{=t}\rho_0) \\
    &\leq  1 - \frac{q^m - q^t}{q^m - 1} \\
    &\leq q^{t-m}.\qedhere
\end{align*}
We now prove the second inequality. By the triangle inequality and the fact that we are working with PSD operators, we have
\[
\norm{\Pi_{<t}(\rho_0 - \rho_1) \Pi_{<t}}_1
\leq \tr(\Pi_{<t}\rho_0) + \tr(  \Pi_{<t} \rho_1).
\]
For $\rho_0$, $t$ columns are drawn uniformly from $\F_q^m$. The probability they are rank-deficient is at most $q^{t-m}$. For $\rho_1$, after conditioning on $v$, the probability that $t$ columns drawn uniformly from $H_v$ are rank-deficient is at most $q^{t-m+1}$. Thus the overall bound is $q^{t-m} + q^{t-m+1} = (q+1) \cdot q^{t-m}$.
\end{proof}

An immediate corollary of the lemma above is the following.

\begin{proposition}
\label{prop:td-bound-deg-2}
We can upper bound the trace norm between $\rho_0,\rho_1$ as
\begin{equation*}
    \norm{\rho_0 - \rho_1}_1 \leq (q+2) \cdot q^{t-m}.
\end{equation*}
\end{proposition}

\cref{prop:td-bound-deg-2} implies an $\Omega(m)$ lower bound via a standard argument.

\begin{theorem}
    Any $\eps$-tester for $m$-variate degree-$2$ phase states requires at least $m-2$ copies for~$\varepsilon\in (0,1-1/q]$. 
\end{theorem}
\begin{proof}
    Suppose there exists an $\eps$-tester for degree-$2$ phase states using $t$ copies. 
    Then there is a POVM element $M$ with $0 \leq M \leq I$ acting on $t$ copies such that for all pure states $\ket{\psi}$: 
    \begin{itemize}
        \item If $\ket{\psi}$ is a degree-$2$ phase state then $\tr(M \ketbra{\psi}{\psi}^{\otimes t}) \geq \frac{2}{3}$.
    \item If $\ket{\psi}$ is $\eps$-far from all degree-$2$ phase states then $\tr(M \ketbra{\psi}{\psi}^{\otimes t}) \leq \frac{1}{3}$. 
    \end{itemize}
Therefore, for any degree-$2$ phase state $\ket{\psi}$ and any state $\ket{\phi}$ that is $\eps$-far from all degree-$2$ phase~states, 
\[
\tr\left(M \left(\ketbra{\psi}{\psi}^{\otimes t} - \ketbra{\phi}{\phi}^{\otimes t} \right) \right) \geq \frac{1}{3}.
\]

By linearity, the same inequality holds if we replace $\ket{\psi}$ with an average over degree-$2$ phase states and $\ket{\phi}$ with an average over states $\eps$-far from all degree-$2$ phase states.
Let $\rho_0$ and $\rho_1$ denote the ensembles defined in \cref{prop:td-bound-deg-2}. Since every pure state from the $\rho_1$ ensemble has an inner product of at most $1/\sqrt{q}$ with any pure state ensemble from $\rho_0$, we see that any sample from $\rho_1$ is $\eps$-far from $\rho_0$. Thus, we find that 
\begin{align*}
    \frac{1}{3} &\leq \tr\left( M \left(\rho_0 - \rho_1 \right) \right) 
    \leq \frac{1}{2}\norm{\rho_0 - \rho_1}_1 \leq \frac{q+2}{2} \cdot q^{t-m}.
\end{align*}
The second inequality follows from the variational characterization of the trace distance, and the third inequality from \cref{prop:td-bound-deg-2}.
Since $q^{t-m} \to 0$ as $n \to \infty$ unless $t \geq \left\lceil m + \log_q\left(\tfrac{2}{3(q+2)}\right) \right\rceil \ge m-2$,
we conclude that no $\eps$-tester exists unless $t \geq m-2$.
\end{proof}

\section{Learning codeword states}
In this section, we \emph{characterize} the sample complexity of learning  $C$-codeword states for a linear code $C \subseteq \F_q^n$  with generator matrix $G$, in terms of the R\'enyi-$1/2$ entropy of the distribution over syndromes $G\be_t$. 
To this end, we first observe that the so-called \emph{pretty-good measurement} (PGM) is the optimal measurement for this learning task, and then analyze the success probability of the~PGM. 

Before proceeding, we briefly compare our results with generic quantum state learning algorithms. Since the family of $C$-codeword states has size $|C|=q^{\dim(C)}$, shadow tomography~\cite{aaronson2018shadow,buadescu2021improved} implies that an unknown codeword state can be identified using $\poly(\dim(C))$ copies. Such guarantees are not always optimal. For example, \cite{arunachalam2022phase} showed that $\RM_2[m,d]$-codeword phase states can be learned using $\Theta(m^{d-1})$ copies, whereas shadow tomography yields only an $O(m^d)$-copy guarantee. Our goal in this section is to characterize the optimal sample complexity of this learning problem.

\subsection{Pretty-good measurement}
Consider an ensemble of quantum states, $\calE=\{(p_i,\ket{\phi_i})\}_{i\in[m]}$, where $p=\{p_1,\ldots,p_m\}$ is a probability distribution. In the quantum learning problem, an algorithm is given a state $\ket{\phi_i}\in \mathcal{E}$ drawn according $p$, and the goal is to identify $i$ with probability $\geq 2/3$. More generally, the goal of the algorithm is to maximize the average success probability to correctly identifying $i$. 

For a POVM $\calM = \{M_i\}_{i \in [m]}$ consisting of positive semidefinite  matrices, the probability of obtaining outcome $j$ on input $\ket{\phi_i}$ is $\braketbra{\phi_i}{M_j}{\phi_i}$. The average success probability is 
\[
P_{\calM}(\mathcal{E}) = \sum_{i=1}^m p_i\braketbra{\phi_i}{M_i}{\phi_i}.
\]
Let $P^{\rm opt}(\mathcal{E})=\max_{\calM} P_{\calM}(\mathcal{E})$ denote the optimal average success probability, where the maximization is taken over all $m$-outcome POVMs. 
For every ensemble $\mathcal{E}$, the PGM is a specific POVM (depending on $\mathcal{E}$) that does \emph{reasonably} well against $\mathcal{E}$. 
In particular, it is well-known that 
\begin{align}
\label{eq:pgmquadratic}
P^{\rm opt}(\mathcal{E})^2 \leq P^{\rm PGM}(\mathcal{E})\leq  P^{\rm opt}(\mathcal{E}).
\end{align}
We now define the POVM elements of the PGM. Let $\ket{\phi'_i}=\sqrt{p_i}\ket{\phi_i}$, and $\mathcal{E}'=\{\ket{\phi'_i} : i\in [m]\}$ be the set of states in~$\mathcal{E}$, renormalized to reflect their probabilities. Define $\rho_{\rm avg}=\sum_{i\in [m]} \ketbra{\phi'_i}{\phi'_i}$. The PGM is defined as the set of measurement operators $\{\ketbra{\nu_i}{\nu_i}\}_{i\in [m]}$ where $\ket{\nu_i}=(\rho_{\rm avg})^{-1/2}\ket{\phi'_i}$ (the inverse square root of $\rho$ is taken over its non-zero eigenvalues). 

Although the PGM is in general only quadratically related to the optimal measurement in terms of the success probability, 
Eldar and Forney~\cite{eldar2002quantum} proved that it is in fact optimal for geometrically uniform ensembles generated by finite abelian unitary groups.

\begin{definition}\label{def:gu-ensemble}
    A geometrically uniform (GU) ensemble is a collection of pure states $\{\ket{\phi_g} = U_g \ket{\phi} : g \in G\}$, where $G$ is a finite abelian group.
\end{definition}

The original proof in \cite{eldar2002quantum} proceeds by a case analysis that can be somewhat difficult to follow. A simpler proof was later given by Zhou, Chessa, Chitambar, and Leditzky~\cite{zhou2025distinguishability}; however, their argument relies on convex optimization theory. Here we present an elementary proof of the same result in \cref{sec:eldar-forney}.

\begin{restatable}[Eldar-Forney {\cite{eldar2002quantum}}]{lemma}{eldarforney}
\label{prop:eldar-forney}
Let $\mathcal{E}=\{\ket{\phi_g}\}_{g \in G}$ be a geometrically uniform ensemble, and consider its average state $\rho_{\rm avg} \coloneqq \frac{1}{|G|} \sum_{g \in G} \ketbra{\phi_g}{\phi_g}$.
Given a uniformly random $\ket{\phi_g}$, the optimal success probability for identifying $g \in G$ is
    \begin{equation*}
        P^{\mathrm{opt}} = \frac{1}{|G|}\bigl(\tr \sqrt{\rho_{\rm avg}}\bigr)^2,
    \end{equation*}
    and it is achieved by the pretty-good measurement on the ensemble $\mathcal{E}$.
\end{restatable}

\subsection{Sample complexity of learning codeword states}

Recall the learning problem: let $C \subseteq \F_q^n$ be a linear code with generator matrix
$G \in \F_q^{k\times n}$. Let $\{\ket{\psi_c} : c \in C\}$ denote the associated codeword-state ensemble, and suppose a learning algorithm is given $t$ copies of
$\ket{\psi_{\bc}}$ for an unknown uniformly random $\bc \in C$; the goal is to identify $\bc$.

Let $U_k$ denote the uniform distribution over $\F_q^k$.
Throughout, let $\bu \sim U_k$ and let $\bc(\bu) \coloneqq \bu G$.
For a discrete probability distribution $p=(p_i)_i$, denote the R\'enyi entropy of order
$1/2$ by
\[
H_{1/2}(p) \coloneqq 2\log_q\!\left(\sum_i \sqrt{p_i}\right).
\]

\begin{theorem}\label{thm:pgm-learning}
Given $t$ copies of a state uniformly chosen from the ensemble
$\{\ket{\psi_{\bc}} : \bc \in C\}$, the optimal success probability for identifying $\bc$, which is achieved by the pretty-good measurement, is
\[
P^{\rm opt}
= \frac{1}{q^k}\left(\sum_{s \in \F_q^k}\sqrt{\Pr[G \be_t = s]}\right)^2
= q^{H_{1/2}(G \be_t)-k},
\]
where $\be_t \sim \calE_{n,q}^t$ and $k$ is the dimension of the code.

Equivalently, if $\mathbf{c}^\perp \in C^\perp$ is a uniformly chosen codeword from the dual code,
and $\bX \coloneqq \mathbf{c}^\perp + \be_t$, then
\[
P^{\rm opt} = \mathrm{BC}(\bX,U_n)^2,
\]
where $U_n$ denotes the uniform distribution on $\F_q^n$, and
$\mathrm{BC}(\cdot,\cdot)$ is the Bhattacharyya coefficient.
\end{theorem}

\begin{proof}
First observe that since $C$ is a linear code, the set of codewords forms an additive
subgroup of $\F_q^n$, and hence the corresponding unitaries
$\{Z^c : c \in C\}$ form a finite abelian group. Therefore the ensemble
$\{\ket{\psi_c}^{\otimes t} : c \in C\}$ is geometrically uniform
(\cref{def:gu-ensemble}), so by \cref{prop:eldar-forney} the pretty-good measurement
is \emph{optimal} for identifying $c \in C$.
In particular, let $\rho_{\rm avg} \coloneqq \frac{1}{|C|}\sum_{c \in C} \ketbra{\psi_c}{\psi_c}^{\otimes t}$. 
Then \cref{prop:eldar-forney} gives
\[
P^{\rm opt}
= \frac{1}{|C|}\bigl(\tr \sqrt{\rho_{\rm avg}}\bigr)^2.
\]
As shown in \cref{prop:spectrum-formula},
\[
\rho_{\rm avg}
= \calZ_C^t(\ketbra{+}{+}^{\otimes t})
= \sum_{s \in \F_q^k} \Pr[G\be_t=s]\ketbra{\phi_s}{\phi_s},
\]
where $\ket{\phi_s} \coloneqq \Pi_s \ket{+}^{\otimes t}/\norm{\Pi_s \ket{+}^{\otimes t}}_2$.
Therefore, 
\begin{align}
\label{eq:successofpgmforlearning-qary}
P^{\rm opt}
&= \frac{1}{|C|}\bigl(\tr \sqrt{\rho_{\rm avg}}\bigr)^2 = \frac{1}{q^k}
\left(\sum_{s \in \F_q^k}\sqrt{\Pr[G\be_t=s]}\right)^2.
\end{align}
This proves the first part of the theorem.

For the second part, let $\mathbf{c}^\perp$ be uniformly distributed on $C^\perp$, and define
$\bX \coloneqq \bc^\perp+\be_t$. Then for every $x \in \F_q^n$,
\begin{align*}
\Pr[\bX=x]
= \frac{1}{|C^\perp|}\sum_{u \in C^\perp}\Pr[\be_t=u+x] 
= \frac{1}{|C^\perp|}\sum_{y \in x+C^\perp}\Pr[\be_t=y] 
= \frac{1}{|C^\perp|}\Pr[\be_t \in x+C^\perp].
\end{align*}
Since $Gx=Gy$ iff $x-y \in C^\perp = \ker G$, the event $\be_t \in x+C^\perp$ is equivalent to $G\be_t = Gx$. 
Using $|C^\perp|=q^{n-k}$, we obtain
\[
\Pr[X=x] = q^{-(n-k)}\Pr[G\be_t=Gx].
\]

Now compute the Bhattacharyya distance:
\begin{align*}
\mathrm{BC}(\bX,U_n)
&= \sum_{x \in \F_q^n}\sqrt{\Pr[\bX=x]\cdot q^{-n}} \\
&= \sum_{x \in \F_q^n} q^{-n+k/2}\sqrt{\Pr[G\be_t=Gx]} \\
&= \sum_{s \in \F_q^k}\ \sum_{x:Gx=s}
   q^{-n+k/2}\sqrt{\Pr[G\be_t=s]} \\
&= \sum_{s \in \F_q^k} q^{n-k}\, q^{-n+k/2}\sqrt{\Pr[G\be_t=s]} \\
&= \sum_{s \in \F_q^k} q^{-k/2}\sqrt{\Pr[G\be_t=s]}.
\end{align*}
Here we grouped terms according to the syndrome $Gx=s$. Since $G$ has rank $k$, every nonempty fiber $\{x \in \F_q^n : Gx=s\}$ is a coset of $\ker G=C^\perp$ and
therefore has size $|C^\perp|=q^{n-k}$.
Squaring the final expression and comparing with
Eq.~\eqref{eq:successofpgmforlearning-qary} yields $P^{\rm opt} = \mathrm{BC}(X,U_n)^2$, as claimed.
\end{proof}

\cref{thm:pgm-learning} computes the optimal success probability when the input state is drawn according to the uniform distribution. The next corollary shows that by symmetry the same value is also the worst-case success probability.

\begin{corollary}\label{cor:pgm-learning-worst-case}
Let $\mathcal{M}_{\rm PGM}$ denote the pretty-good measurement for the uniform ensemble
$\{\ket{\psi_c}^{\otimes t} : c \in C\}$.
Then for every $c \in C$, $\Pr[\mathcal{M}_{\rm PGM}\text{ outputs }c\mid \ket{\psi_c}^{\otimes t}] = q^{H_{1/2}(G\be_t)-k}$.
We additionally have $$\max_{\mathcal{M}} \ \min_{c \in C} \Pr[\mathcal{M}\text{ outputs }c \mid \ket{\psi_c}^{\otimes t}] = q^{H_{1/2}(G\be_t)-k}.
$$
In particular, there exists a measurement that identifies every $c \in C$ with success
probability at least $1-\delta$ if and only if $H_{1/2}(G\be_t)\ge k + \log_q(1-\delta)$.
\end{corollary}
\begin{proof}
Define
\[
\ket{\phi_c} \coloneqq q^{-k/2}\ket{\psi_c}^{\otimes t},
\qquad
\rho_{\rm avg} \coloneqq \sum_{c \in C}\ketbra{\phi_c}{\phi_c}
= \frac{1}{|C|}\sum_{c \in C}\ketbra{\psi_c}{\psi_c}^{\otimes t}.
\]
Let $\mathcal{M}_{\rm PGM}=\{\ketbra{\nu_c}{\nu_c}\}_{c \in C}$ be the pretty-good measurement
for this ensemble, where $\ket{\nu_c} = \rho_{\rm avg}^{-1/2}\ket{\phi_c}$.

For each $a \in C$, let $V_a \coloneqq (Z^a)^{\otimes t}$.
Since $\ket{\psi_{c+a}}^{\otimes t} = V_a \ket{\psi_c}^{\otimes t}$, we also have $V_a\ket{\phi_c} = \ket{\phi_{c+a}}$.
Therefore,
\[
V_a \rho_{\rm avg} V_a^\dagger
= \sum_{c \in C} \ketbra{\phi_{c+a}}{\phi_{c+a}}
= \rho_{\rm avg}.
\]
It follows that $V_a$ commutes with $\rho_{\rm avg}^{-1/2}$ on the support of $\rho_{\rm avg}$, and hence
\[
V_a\ket{\nu_c}
= V_a \rho_{\rm avg}^{-1/2}\ket{\phi_c}
= \rho_{\rm avg}^{-1/2}V_a\ket{\phi_c}
= \rho_{\rm avg}^{-1/2}\ket{\phi_{c+a}}
= \ket{\nu_{c+a}}.
\]
Thus the PGM projectors satisfy $\ketbra{\nu_{c+a}}{\nu_{c+a}} = V_a \ketbra{\nu_c}{\nu_c} V_a^\dagger$.

Now define
\[
p_c \coloneqq
\Pr[\mathcal{M}_{\rm PGM}\text{ outputs }c \mid \ket{\psi_c}^{\otimes t}]
= \braket{\psi_c^{\otimes t}|\nu_c\rangle\!\langle\nu_c}{\psi_c^{\otimes t}}.
\]
For every $a,c \in C$,
$p_{c+a}
= \braket{\psi_{c+a}^{\otimes t}| \nu_{c+a}\rangle\!\langle \nu_{c+a}}{\psi_{c+a}^{\otimes t}}  = \braket{\psi_c^{\otimes t}|\nu_c\rangle\!\langle\nu_c}{\psi_c^{\otimes t}}$,
so $p_c$ is independent of $c$.

On the other hand, \cref{thm:pgm-learning} shows that the average success probability of
$\mathcal{M}_{\rm PGM}$ under the uniform distribution on $C$ is
\[
\frac{1}{|C|}\sum_{c \in C} p_c
= q^{H_{1/2}(G\be_t)-k}.
\]
Since all the terms $p_c$ are equal, each of them must equal
$q^{H_{1/2}(G\be_t)-k}$. This proves the first claim. 

For the worst-case statement, let $\mathcal{M}$ be any POVM, and define $q_c(\mathcal{M}) \coloneqq \Pr[\mathcal{M}\text{ outputs }c \mid \ket{\psi_c}^{\otimes t}]$.
Then
\[
\min_{c \in C} q_c(\mathcal{M})
\le
\frac{1}{|C|}\sum_{c \in C} q_c(\mathcal{M}).
\]
By \cref{thm:pgm-learning}, the right-hand side is at most
$q^{H_{1/2}(G\be_t)-k}$, since that theorem gives the optimal average success probability
for the uniform ensemble. Therefore,
\[
\max_{\mathcal{M}} \ \min_{c \in C} q_c(\mathcal{M})
\le
q^{H_{1/2}(G\be_t)-k}.
\]
Conversely, the pretty-good measurement satisfies
\[
q_c(\mathcal{M}_{\rm PGM})=q^{H_{1/2}(G\be_t)-k}
\qquad \text{for every } c \in C,
\]
so $\min_{c \in C} q_c(\mathcal{M}_{\rm PGM})
= q^{H_{1/2}(G\be_t)-k}$.
Hence
\[
\max_{\mathcal{M}} \ \min_{c \in C}
\Pr[\mathcal{M}\text{ outputs }c \mid \ket{\psi_c}^{\otimes t}]
=
q^{H_{1/2}(G\be_t)-k}.
\]

The final claim follows immediately by comparing this quantity with $1-\delta$ and taking
$\log_q$.
\end{proof}

\section{Correcting random errors in \texorpdfstring{$q$}{q}-ary Reed--Muller codes}
\label{sec:asw}

In this section, we generalize a result due to Abbe, Shpilka, and Wigderson~\cite{ASW15} to arbitrary finite fields to fill in the pending~\cref{corollary:rm-decode-bounds} leveraged in our low-degree phase state testing lower bounds (\cref{subsec:low-deg-phase-testing-lbs}).
In particular, we study the ability of Reed--Muller codes over $\F_q$ to correct uniformly random error patterns. Our proof structure closely follows that of \cite{ASW15}, while also streamlining some aspects of the argument. However, the proof of~\cite{ASW15} relies on a result of Wei~\cite{wei2002generalized} regarding generalized Hamming weights of the binary Reed--Muller code, which does not directly generalize the $q$-ary Reed--Muller code. As a result, we instead rely on a result of Heijnen and Pellikaan~\cite{heijnen1997generalized} on the generalized Hamming weights of $\RM_q[m,d]$ as a substitute of~\cite{wei2002generalized}.\medskip

Let $D_{d,q}(r)$ denote the number of monomials in $r$ variables of total degree at most $d$, with each variable having degree at most $q-1$. The main theorem proved in this section is the following.

\begin{restatable}[Generalization of {\cite[Theorem 6.2]{ASW15}}]{theorem}{rmcorrect}
\label{thm:asw-62}
Let $q$ be a prime number.
Let $m \geq 1$, $0 \leq d \leq \lfloor \tfrac{m(q-1)-2}{2} \rfloor$, and $0 \leq r \leq m$.
Let $t \leq D_{d,q}(r)$ be a positive integer, and let $\be_t = \sum_{i=1}^t \bc_i \cdot \delta_{\ba_i}$ for uniformly random $\bc_i \in \F_q^\times$ and $\ba_i \in \F_q^m$.
Then $\RM_q[m, m(q-1) - (2d + 2)]$ can correct $\be_t$ with probability at least $1 - t q^{r-m}$.
\end{restatable}

\subsection{On random submatrices of the generator matrix}

We generalize \cite[Theorem 4.5]{ASW15} to arbitrary finite fields. 
In short, the theorem states that submatrices obtained by sampling random columns of the generator matrix of $\RM_q[m,d]$ have full column rank with high probability.
We begin by introducing the necessary definitions and notation.

\begin{definition}[Support of a code]
For a linear code $D$, define $\supp(D) \coloneqq \{i : \exists y \in D \text{ such that $y_i \neq 0$}\}$.
That is, $D$ is the union of the supports of all the codewords in $D$.   
\end{definition}

\begin{definition}[Generalized Hamming weight]
\label{def:gen-ham-wei}
For $1 \leq a \leq \dim C$, let $d_a(C)$ denote the $a$'th generalized Hamming weight of $C$, which is defined as 
\[
d_a(C) = \min\{\abs{\supp(D)} : D \leq C, \dim D = a\}.
\]
That is, $d_a(C)$ is the size of the smallest set of coordinates for which there exists a linear subcode $D$ with dimension $a$.
\end{definition}

Let $\calP_{m,d,q} \coloneqq \Span_{\F_q}\{ x^\alpha : \alpha \in \{0, \dots, q-1\}^m, \abs{\alpha} \leq d\}$, viewed as a vector space of functions $\F_q^m \to \F_q$. Equivalently, $\calP_{m,d,q}$ is the subspace of degree $\le d$ polynomials in $\F_q[x_1, \dots, x_m]/(x_1^q -x_1, \dots, x_m^q - x_m)$.
Define 
\[
D_{d,q}(r) \coloneqq \dim \calP_{r,d,q} =
\abs{\{\alpha=(\alpha_1,\dots,\alpha_r)\in \{0,1,\dots,q-1\}^r: |\alpha|\coloneqq \alpha_1+\cdots+\alpha_r\le d\}}.
\]

For $x \in \F_q^m$, define $E_x : \calP_{m,d,q} \to \F_q$ as $E_x(f) \coloneqq f(x)$, and let $\calM_{m,d,q}$ denote the monomial basis of $\calP_{m,d,q}$. 
With respect to the monomial basis $\calM_{m,d,q}$, the functional $E_x$ is represented by $(g(x))_{g \in \calM_{m,d,q}}$, which is exactly the column of the generator matrix of $\RM_q[m,d]$ indexed by $x$.
Hence, linear independence of $E_{v_1},\dots,E_{v_t}$ is equivalent to linear independence of the corresponding columns of the generator matrix of $\RM_q[m,d]$.

\begin{theorem}[Generalization {\cite[Theorem 4.5]{ASW15}}]
\label{thm:asw-45}
Let $q$ be a prime number.
Let $m \geq 1$, $0 \leq d \leq m(q-1)$, and $0 \leq r \leq m$.
Let $v_1, \dots, v_t$ be independent uniformly random points of $\F_q^m$.
If $t \leq D_{d,q}(r)$, then 
\[
\Pr[E_{v_1}, \dots, E_{v_t}\text{ are linearly independent}] \geq 1 - t q^{r-m}.
\]
\end{theorem}

To recover \cite[Theorem 4.5]{ASW15}, specialize to $\F_2$ and choose $r$ so that $1 - t q^{r-m} \geq 1 - \eps$.

\begin{proof}
If $d = 0$, then $\calP_{m,0,q}$ is the constant polynomials and $D_{0,q}(r) = 1$ for every $r$.
Hence, $t \leq 1$ by hypothesis, and so $E_{v_1}, \dots, E_{v_t}$ is linearly independent. Thus, the theorem is true for $d=0$. Henceforth, assume $1 \leq d \leq m(q-1)$.

For $A \subseteq \F_q^m$, define $\punc_A: \calP_{m,d,q} \to \F_q^A$ as $\punc_A(f) \coloneqq \left.f\right|_A$, the map that punctures/restricts polynomials to the points in $A$.
Define 
\[
I_d(A) \coloneqq \ker(\punc_A) = \{f \in \calP_{m,d,q} : f(a) = 0 \text{ for all $a \in A$}\}, 
\]
and $h_d(A) \coloneqq \rank(\punc_A)$. By rank-nullity, $h_d(A) = \dim \calP_{m,d,q} - \dim I_d(A)$.
(From a coding-theoretic perspective, $h_d(A)$ is the dimension of the Reed--Muller code $\RM_q[m,d]$ punctured to the coordinates in $A$.)
We also have the following fact. 

\begin{fact}
\label{fact:asw-ker-dim}
    $h_d(A) = \dim \Span \{E_a : a \in A\}$.
\end{fact}
\begin{proof}
Let $\punc_A^*: (\F_q^A)^* \to \calP_{m,d,q}^*$ be the dual map of $\punc_A$, which, by definition, is given by $\punc_A^*(\lambda) = \lambda \circ \punc_A$.
Let $\delta_a \in (\F_q^A)^*$ be the coordinate functional, where $\delta_a(x) = x_a$.
For $a \in A$ and $f \in \mathcal P_{m,k,q}$,
\[
\punc_A^*(\delta_a)(f)
= \delta_a(\punc_A(f)) = \delta_a(f|_A) = f(a) = E_a(f),
\]
so $\punc_A^*(\delta_a)=E_a$.
Therefore, $\im(\punc_A^*) = \Span\{E_a : a \in A\}$. 
Since $\rank(\punc_A^*) = \rank(\punc_A) = h_d(A)$, the result follows.
\end{proof}
Finally, define the closure of $A$ as $\cl_d(A) \coloneqq \{x \in \F_q^m : f(x) = 0 \text{ for every $f \in I_d(A)$} \}$.  We need the following three lemmas to prove the main theorem.

\begin{lemma}
\label{lemma:closure-lem}
For every $A \subseteq \F_q^m$ and every $x \in \F_q^m$, 
\[
x \in \cl_d(A) \iff E_x \in \Span\{E_a : a \in A\}.
\]
\end{lemma}
\begin{proof}
We have 
\begin{align*}
I_d(A) 
&= \{f : f(a) = 0 \text{ for all $a \in A$}\}\\
&= \{f : E_a(f) = 0 \text{ for all $a \in A$}\}\\
&= \bigcap_{a \in A} \ker(E_a). 
\end{align*}
Hence, for any $x \in \F_q^m$, we have
\begin{align*}
x \in \cl_d(A) 
\iff f(x) = 0 \text{ for all $f \in I_d(A)$} 
\iff E_x(f) = 0 \text{ for all $f \in \bigcap_{a \in A}\ker(E_a)$}.
\end{align*}
Suppose that $E_x \in \Span\{E_a : a\in A\}$. Then we may write $E_x = \sum_{a\in A} c_a E_a$ for some coefficients $c_a \in \F_q$. If $f\in \bigcap_{a\in A}\ker(E_a)$, then
$E_a(f)=0$ for every $a\in A$, so $E_x(f) = 0$. Thus $x\in \cl_d(A)$.

We prove the other direction by contraposition. Suppose $x \notin \cl_d(A)$. Then there exists $f \in I_d(A)$ such that $f(x)\neq 0$, i.e., $f(a)=0$ for all $a \in A$ but $f(x)\neq 0$. If $E_x \in \Span\{E_a : a\in A\}$, then we may write $E_x = \sum_{a\in A} c_a E_a$ for some coefficients $c_a$. 
We then have 
\[
0 \neq f(x) = E_x(f) = \sum_{a \in A} c_a E_a(f) = \sum_{a \in A} c_a f(a) = 0, 
\]
a contradiction.
Hence $E_x \notin \Span\{E_a : a\in A\}$.
\end{proof}

\begin{lemma}
\label{lemma:ideals-dont-change}
For every $A \subseteq \F_q^m$, we have $I_d(\cl_d(A)) = I_d(A)$, and therefore $h_d(\cl_d(A)) = h_d(A)$.
\end{lemma}
\begin{proof}
Since $A \subseteq \cl_d(A)$, any polynomial that vanishes on $\cl_d(A)$ also  vanishes on $A$. Hence $I_d(\cl_d(A)) \subseteq I_d(A)$.
On the other hand,
if $f \in I_d(A)$, then, by definition, every point of $\cl_d(A)$ is a common zero of all polynomials in $I_d(A)$, so $f$ vanishes on $\cl_d(A)$. Thus, $f \in I_d(\cl_d(A))$, and $I_d(A) \subseteq I_d(\cl_d(A))$. 
The equality $h_d(\cl_d(A)) = h_d(A)$ follows immediately from the definition of~$h_d$.
\end{proof}

\begin{lemma}
\label{lemma:generalized-hamming-weights}
Let $d \geq 1$. If $A \subseteq \F_q^m$ satisfies $h_d(A) \leq D_{d,q}(r)$, then $\abs{A} \leq q^r$.
\end{lemma}
\begin{proof}
If $r = m$, then $\abs{A} \leq q^m = q^r$ is trivial.
If $d = m(q-1)$, then $\calP_{m,d,q}$ is the full space of all functions $\F_q^m \to \F_q$. In that case, $h_d(A) = \abs{A}$. Because every vector in $\{0,\dots, q-1\}^r$ has coordinate sum at most $r(q-1) \leq m(q-1) = d$, we can also conclude that $D_{d,q}(r) = q^r$. Thus, $\abs{A} = h_d(A) \leq D_{d,q}(r) = q^r$.
Henceforth, we assume $1 \leq d < m(q-1)$ and $r < m$.

Define $C \coloneqq \RM_q[m,d]$. It is clear $\dim C = \dim \calP_{m,d,q} = D_{d,q}(m)$.
For convenience, set $N \coloneqq D_{d,q}(m)$ and $a \coloneqq N - D_{d,q}(r)$. Because $d \geq 1$ and $r < m$, we have $a \geq 1$.
We will prove that $d_a(C) = q^m - q^r$, where $d_a$ is the generalized Hamming weight (\cref{def:gen-ham-wei}).
To do so, we will use the following theorem of Heijnen and Pellikaan~\cite{heijnen1997generalized}. Indeed, this lemma is a basic corollary of \cite{heijnen1997generalized}, with details included for completeness. 

\begin{theorem}[{\cite[Theorem 5.10]{heijnen1997generalized}}]
\label{thm:generalized-wei}
Let $Q \coloneqq \{0,1,\dots, q-1\}$.
Let $\alpha$ be the $a$'th element of $Q^m$ in lexicographic order with the property that $\sum_i \alpha_i > (q-1)m - d - 1$. Then 
\[
d_a(\RM_q[m,d]) = \sum_{i = 1}^m \alpha_{i}q^{m-i} + 1.
\]
\end{theorem}
We will construct an $\alpha^*$ that produces our claimed identity.
Define $d^\perp \coloneqq (q-1)m - d -1$ and $F_{>d^\perp} \coloneqq \{\alpha \in Q^m : \abs{\alpha} > d^\perp\}$.
\cref{thm:generalized-wei} says that if $\alpha$ is the $a$'th element of $F_{>d^\perp}$ in lexicographic order, then $d_a(C) = \sum_{j=1}^m \alpha_j q^{m-i} + 1$.
Define also $B_{\leq d} \coloneqq \{x \in Q^m : \abs{x} \leq d\}$ and $\mu \coloneqq (q-1, \dots, q-1) \in Q^m$.
Observe that $\abs{B_{\le d}} = N$.

Consider the map $\beta \mapsto \mu - \beta$.
We have $\abs{\mu - \beta} = m(q-1) - \abs{\beta}$. So, $\abs{\beta}\leq d \iff \abs{\mu - \beta} > d^\perp$.
Thus the map is a bijection between $B_{\le d}$ and $F_{> d^\perp}$. It is also easy to see that the map reverses lexicographic order.

Recall that $D_{d,q}(r)$ is the number of monomials in $r$ variables of degree at most $d$. Thus the first $D_{d,q}(r)$ elements of $B_{\le d}$ are vectors of the form 
$(0,\ldots, 0, \beta_{m-r+1}, \dots, \beta_m)$ with $\sum_i \beta_i \leq d$.
Therefore, the \((D_{d,q}(r)+1)\)-st element of $B_{\le d}$ in lexicographic order is $\beta^*=(0,\dots,0,1,0,\dots,0)$,  where the $1$ occurs in coordinate $m-r$.

Because $\abs{B_{\le d}}=N$ and the map $\beta\mapsto \mu-\beta$ reverses order, it follows that
\[
\alpha^* \coloneqq \mu-\beta^* = (q-1,\dots,q-1,q-2,q-1,\dots,q-1)
\]
is the $N-(D_{d,q}(r)+1)+1=N-D_{d,q}(r)=a$'th element of $F_{>d^\perp}$.
Thus, by applying \cref{thm:generalized-wei} with $\alpha^*$, we get 
\begin{align*}
    d_a(C) = \left((q-1) \sum_{i=1}^m q^{m-i}\right) - (q-1)q^r + (q-2) q^r + 1
    = q^m - q^r,  
\end{align*}
where we use the identity $\sum_{i=1}^m q^{m-i} = \tfrac{q^m - 1}{q-1}$.

Recall that $\punc_A: C \to \F_q^A$ denotes the map that punctures $C$ to the coordinates in $A$. Recall that $\dim \im(\punc_A) = h_d(A)$.
By our hypothesis, $\dim \ker(\punc_A) = N - h_d(A) \geq N - D_{d,q}(r) = a$.
Let $D \leq \ker(\punc_A)$ be an $a$-dimensional subcode of $C$.
Thus, every codeword in $D$ is $0$ on the coordinates of $A$, so $\supp(D) \subseteq A^c$. Then $\abs{A^c} \geq \abs{\supp(D)} \geq d_a(C) = q^m - q^r$.
Therefore, $\abs{A}\le q^r$.
\end{proof}

We are now ready to prove the theorem. 
For $0 \le j \le t$, let $\calE_j$ be the event that the vectors $E_{v_1}, \dots, E_{v_j}$ are linearly independent. (Note $\calE_0$ is thus the whole probability space.) For $0 \le j \le t$, let $A_j \coloneqq \{v_1, \dots, v_j \} \subseteq \F_q^m$.

Fix some $j \in \{0,1,\ldots , t\}$. On the event $\calE_j$, the functionals $E_{v_1}, \dots, E_{v_j}$ are linearly independent. Because $\Span\{E_a : a \in A_j\} = \Span\{E_{v_1}, \dots, E_{v_j}\}$ by definition, \cref{fact:asw-ker-dim} gives $h_d(A_j) = j$.
By \cref{lemma:ideals-dont-change}, $h_d(\cl_d(A_j)) = j$.
Then, because $j < t \leq D_{d,q}(r)$, \cref{lemma:generalized-hamming-weights} states $\abs{\cl_d(A_j)} \leq q^r$.

Conditioning on event $\calE_j$, the event $\calE_{j+1}^c$ is the event that 
\[
\{E_{v_{j+1}} \in \Span\{E_{v_1}, \dots, E_{v_j}\}\}  
= 
\{E_{v_{j+1}} \in \Span\{E_a : a \in A_j\}\}
=  \{v_{j+1} \in \cl_d(A_j)\}. 
\]
The first equality follows from the fact that $A_j = \{v_1,\dots, v_j\}$ and the second equality follows from \cref{lemma:closure-lem}.
Because $v_{j+1}$ is independent of $v_1, \dots, v_j$ and is uniform on $\F_q^m$, we obtain, conditioned on the event $\calE_j$, $\Pr[v_{j+1} \in \cl_d(A_j) | v_1, \dots, v_j] = \tfrac{\abs{\cl_d(A_j)}}{q^m} \leq q^{r - m}$.
Therefore, 
\begin{align*}
    \Pr[\calE_j \cap \calE_{j+1}^c] 
    &= \Pr[\calE_j] \cdot \Pr[\calE^c_{j+1} | \calE_j] \leq \Pr[\calE_j] \cdot q^{r-m} \leq q^{r-m}.
\end{align*}

The event $\calE_j \cap \calE_{j+1}^c$ is the event that the first linear dependence is sampled on the $(j+1)$st sample.
Thus, for $0 \leq j \leq t-1$, these events are pairwise disjoint. Furthermore, their union is the event $\calE_t^c = \{E_{v_1}, \dots, E_{v_t} \text{ are linearly dependent}\}$.
We have 
\begin{align*}
    \Pr[\calE_t^c] = \sum_{j = 0}^{t-1} \Pr[\calE_j \cap \calE_{j+1}^c] \leq \sum_{j=0}^{t-1} q^{r-m} = t q^{r-m}.
\end{align*}
Therefore, $\Pr[\calE_t] \geq 1 - t q^{r-m}$, which establishes the theorem.
\end{proof}

\subsection{Correcting random errors}

We will now use \cref{thm:asw-45} to prove \cref{thm:asw-62}, generalizing \cite[Theorem 6.2]{ASW15}.
We restate the theorem for convenience.

\rmcorrect*

Let $E_q(m,d)$ denote the parity check matrix for $\RM_q[m,m(q-1) - d -1]$.
Recall that the rows of $E_q(m,d)$ are indexed by $\calM_{m,d,q}$ and the columns are indexed by $\F_q^m$.
Moreover, for $u \in \F_q^m$, the $u$'th column of $E_{q}(m,d)$ is given by the vector $(g(u))_{g \in \calM_{m,d,q}}$.

\begin{definition}
For error patterns $\be, \be' \in \F_q^{q^m}$, define the equivalence relation $\sim_d$ as 
\[
\be \sim_d \be' \iff \sum_{x \in \F_q^m} \be(x) f(x) = \sum_{x \in \F_q^m} \be'(x) f(x) \qquad \forall f \in \calP_{m,d,q}.
\]
\end{definition}

The equivalence relation $\sim_d$ captures exactly when two error patterns produce the same syndrome.

\begin{lemma}
\label{lemma:syndrome-to-equiv}
   For $\be, \be' \in \F_q^{q^m}$, 
   \[
   E_q(m,d) \be = E_q(m,d) \be' \iff \be \sim_d \be'.
   \]
\end{lemma}
\begin{proof}
For each monomial $g \in \calM_{m,d,q}$, the $g$'th coordinate of $E_q(m,d)$ is $\sum_{x \in \F_q^m}\be(x) g(x)$. Therefore, $E_q(m,d) \be = E_q(m,d) \be'$ if and only if the sums agree for every $m \in \calM_{m,d,q}$.
Because $\calM_{m,d,q}$ is a basis for $\calP_{m,d,q}$, this is equivalent to the sums agreeing for every $f \in \calP_{m,d,q}$ by linearity. 
\end{proof}

Following \cite{ASW15}, for $U = \{u_1, \dots, u_t\} \subseteq \F_q^m$, let $U^d$ denote the submatrix of $E_q(m,d)$ consisting of the columns indexed by $u_1, \dots, u_t$.
Let $\mathbf{1}_{u_i} \in \F_q^{q^m}$ be the vector that is $1$ in the $u_i$'th coordinate and $0$ elsewhere.

\begin{lemma}[Generalization of {\cite[Lemma 6.11]{ASW15}}]
\label{lemma:indep-gives-uniqueness}
Let $\be = \sum_{i=1}^t c_i \mathbf{1}_{u_i}$ be an error pattern of weight $t$, where $U = \{u_1, \dots, u_t\} \subseteq \F_q^m$ and $c_1, \ldots , c_t \in \F_q^\times$. 
Assume that the columns of $U^d$ are linearly independent. Then for every error pattern $\be'$ of weight $t' \leq t$, we have $\be' \sim_{2d+1} \be \implies \be' = \be$.
In particular, if $\be' \neq \be$ and the weight of $\be'$ is $\leq t$, then $\be' \not\sim_{2d+1} \be$. 
\end{lemma}
\begin{proof}
Define the evaluation map $\eval_U : \calP_{m,d,q} \to \F_q^t$ as $f \mapsto (f(u_1), \dots, f(u_t))$.
Because the columns of $U^d$ are linearly independent, $\eval_U$ is surjective and has rank $t$.
Therefore, there exist polynomials $f_1, \dots, f_t \in \calP_{m,d,q}$ such that $f_i(u_j) = \delta_{ij}$ for all $1 \leq i,j \leq t$.

Let $V = \{v_1, \dots, v_{t'}\} \subseteq \F_q^m$ with $t' \leq t$. Let $\be' = \sum_{j=1}^{t'} c'_j \mathbf{1}_{v_j}$, where $c'_j \in \F_q^\times$.
Assume $\be' \sim_{2d+1} \be$.
For each $i \in [t]$, define $g_i \coloneqq (f_i(v_1), \dots, f_i(v_{t'}))^T$, the column vector containing the evaluations of $f_i$ on $V$. 
Define $W' \coloneqq \diag(c'_1, \dots, c'_{t'})$.

Recall the basic that that for $f,f'\in\calP_{m,d,q}$, the pointwise product $ff'$ is an element of $\calP_{m,2d,q}$. 
Hence, we have 
\begin{align*}
g_i^T W' g_j 
&= \sum_{k = 1}^{t'} c'_k f_i(v_k) f_j(v_k) \\
&= \sum_{k = 1}^{t'} \be'(k) f_i(v_k) f_j(v_k) \\
&= \sum_{k = 1}^{t} \be(k) f_i(u_k) f_j(u_k) \\
&= \sum_{k = 1}^{t} c_k \delta_{ik} \delta_{jk} \\
&= c_i \delta_{ij},
\end{align*}
where the third equality follows from $\be' \sim_{2d+1} \be$. 
We have shown that $g_1,  \dots, g_{t}$ are orthogonal under the bilinear form induced by $W'$, so they are linearly independent. Since $g_i \in \F_q^{t'}$, we must have $t \leq t'$. By assumption $t' \leq t$, so $t' = t$.
Therefore, $g_1, \dots, g_t$ form a basis of $\F_q^t$.

Now for $\ell \in [m]$, define $D_\ell \coloneqq \diag((v_1)_\ell, \dots, (v_t)_\ell)$.
Because $x_\ell ff'$ is an element of $\mathcal P_{m,2d+1,q}$, we can again use the assumption $\be' \sim_{2d+1} \be$ to obtain, for all $i,j \in [t]$,
\[
(D_\ell g_i)^T W' g_j = \sum_{k=1}^t c'_k (v_k)_\ell f_i(v_k) f_j(v_k) = \sum_{k=1}^t c_k (u_k)_\ell f_i(u_k) f_j(u_k) = (u_i)_\ell c_i \delta_{ij}.
\]
Combining these two identities, we get for all $j$
\[
(D_\ell g_i)^T W' g_j = (u_i)_\ell g_i^T W' g_j  
\iff  (D_\ell g_i - (u_i)_\ell g_i)^T W' g_j = 0.
\]
Because $g_1, \dots, g_t$ are linearly independent, the only vector in their span orthogonal to all of them is $0$.
Therefore, $D_\ell g_i = (u_i)_\ell g_i$.
Because we proved $g_i^T W' g_i = c_i$, we know there exists a coordinate $k$ such that $(g_i)_k \neq 0$. On this coordinate, we have $(D_\ell g_i)_k = (v_k)_\ell (g_i)_k = (u_i)_\ell (g_i)_k$. This implies that $(v_k)_\ell = (u_i)_\ell$. Since this holds for all $\ell \in [m]$, we have $v_k = u_i$. 
We have shown that for each index $i \in [t]$, there exists a $k \in [t]$ such that $v_k = u_i$, i.e., $U \subseteq V$. Because $\abs{U} = \abs{V} = t$, we conclude $U = V$.

Finally, relabel the vectors in $V$ so we have $u_i = v_i$. Then, again using $\be' \sim_{2d+1} \be$, we have $c_i = \sum_{j=1}^t c_j f_i(u_j) = \sum_{j=1}^t c'_j f_i(v_j) = c'_i$. 
Thus, $\be' = \be$.
\end{proof}

\cref{lemma:syndrome-to-equiv} essentially shows that if the coordinates of an error pattern are linearly independent, then we can correct that error pattern in $\RM_q[m,m(q-1) - (2d+2)]$. 

\begin{corollary}[Generalization of {\cite[Theorem 6.13]{ASW15}}]
\label{cor:indep-lets-u-decode}
   Let $U \subseteq \F_q^m$ with $\abs{U} = t$, and assume that the columns of $U^d$ are linearly independent. Then every error pattern supported on $U$ is correctable in $\RM_q[m, m(q-1) - (2d+2)]$. 
\end{corollary}
\begin{proof}
Let $\be$ be any error pattern supported on $U$. By \cref{lemma:indep-gives-uniqueness}, there is no error pattern $\be' \neq \be$ of weight at most $t$ such that $\be' \sim_{2d+1} \be$.
\cref{lemma:syndrome-to-equiv} implies there is no such $\be'$ with $E_q(m,2d+1) \be' = E_q(m,2d+1) \be$. 
Because $E_q(m,2d+1)$ is the parity check matrix of $\RM_q[m,m(q-1) - (2d+2)]$, the syndrome of $\be$ uniquely determines $\be$ among all error patterns of weight at most $t$.
\end{proof}

Combining \cref{cor:indep-lets-u-decode,thm:asw-45}  implies \cref{thm:asw-62}.

\begin{proof}[Proof of \cref{thm:asw-62}]
    Let $\be$ be a uniformly random error pattern of weight $t$, and let $U \subseteq \F_q^m$ be the coordinates that support $\be$.
    Because each size-$t$ subset of $\F_q^m$ supports exactly $(q-1)^t$ error patterns of weight $t$, the support of $U$ is uniformly random.

    Let $v_1, \dots, v_t$ be independent uniformly random points of $\F_q^m$. Define the events 
    \[ 
    \calE_1 \coloneqq \{E_{v_1}, \dots, E_{v_t}\text{ are linearly independent}\},\qquad \calE_2 \coloneqq \{v_1, \dots, v_t\text{ are pairwise distinct}\}.
    \]
    Because $\calE_1 \subseteq \calE_2$ and by \cref{thm:asw-45}, we have $\Pr[\calE_1 | \calE_2] = \tfrac{\Pr[\calE_1]}{\Pr[\calE_2]} \geq \Pr[\calE_1] \geq 1 - t q^{r-m}$. 
    Conditioned on $\calE_2$, the ordered tuple $(v_1, \dots, v_t)$ is uniformly distributed over all such tuples of distinct points of $\F_q^m$. Because event $\calE_1$ is invariant under permutation, $\Pr[\calE_1 | \calE_2]$ is precisely the probability that uniformly random size-$t$ subset $U\subseteq \F_q^m$ has $U^d$ with linearly independent columns. When $U^d$ has linearly independent columns, \cref{cor:indep-lets-u-decode} says that every error pattern supported on $U$ can be corrected in $\RM_q[m,m(q-1) - (2d+2)]$. Hence, a uniformly random error pattern of weight $t$ is correctable with probability at least $1 - t q^{r-m}$.
\end{proof}

\cref{thm:asw-62} precisely generalizes \cite[Theorem 6.2]{ASW15}. The following corollary is a related statement that is more useful for our property testing lower bounds in \cref{sec:testing-LBs}.

\begin{restatable}{corollary}{corrmbound}
\label{cor:random-error-unique}
Let $q$ be a prime number, $m \ge 1$, 
$0 \le d \le \lfloor \tfrac{m(q-1)-2}{2}\rfloor$, and $0 \le r \le m$.
Let $t \le D_{d,q}(r)$, and let $\be_t\sim \calE_{q^m,q}^t$.
Let $H$ be the parity check matrix for $\RM_q[m, m(q-1) - (2d + 2)]$.
Then with probability at least $1-tq^{r-m}$, the error pattern $\be_t$ is the unique error pattern of weight at most $t$ having syndrome $H\be_t$.
\end{restatable}
\begin{proof}
Let $U=\{\ba_1,\dots,\ba_t\}$.
By \cref{thm:asw-45}, with probability at least $1-tq^{r-m}$, the evaluation functionals $E_{\ba_1},\dots,E_{\ba_t}$ are linearly independent; equivalently,
the columns of $U^d$ are linearly independent.
On this event, \cref{lemma:indep-gives-uniqueness} implies that $\be_t$ is the unique error pattern of weight at most $t$ in its $\sim_{2d+1}$-class.
By \cref{lemma:syndrome-to-equiv}, this is equivalent to being the unique error pattern of weight at most $t$ with syndrome $H\be_t$.
\end{proof}

Using \cref{cor:random-error-unique}, we bound the success probability of the MAP decoder, which is used directly in our testing lower bounds.

\decodebounds*

\begin{proof}
Note that the dual code of $C$ is $C^\perp = \RM_q[m, m(q-1) - (2d+2)]$. Thus, we can view $G$ as the parity check matrix of $C^\perp$.
Set $\gamma\coloneqq tq^{r-m}$.
By \cref{cor:random-error-unique}, there is a decoder $\dec_{\le t}$ such that
\[
\Pr[\dec_{\le t}(G\be_t)=\be_t]\ge 1-\gamma.
\]
Therefore, by \cref{def:map},
\[
\sqrt{1 - p_{\rm decode}^t(C^\perp)^2}
\le \sqrt{1-(1-\gamma)^2}
\le \sqrt{2\gamma}
= \sqrt{2tq^{r-m}}. \qedhere
\]
\end{proof}

\paragraph{Acknowledgments.} We thank Arkopal Dutt, Yeongwoo Hwang, Amir Shpilka, Nick Hunter-Jones, Jesse Goodman, Vinayak Kumar, and Michael Jabber for helpful conversations. 

\printbibliography

\appendix 
\crefalias{section}{appendix}

\section{Proof of the Eldar--Forney theorem}
\label{sec:eldar-forney}

\eldarforney*

\begin{proof}
    Let $\{N_g\}_{g \in G}$ be any POVM for guessing $g$. Define the group average by
    \begin{equation*}
        M_g \coloneqq \frac{1}{|G|} \sum_{h \in G} U_h^\dagger N_{g+h} U_h.
    \end{equation*}
    Then, by explicit calculation, it is easy to see that $\{M_g\}_{g \in G}$ is again a POVM, and, for every $k \in G$, $M_{g+k} = U_k M_g U_k^\dagger$.
    Its success probability is the same as that of $\{N_g\}$:
    \begin{align*}
        \frac{1}{|G|}\sum_{g \in G}\bra{\phi_g} M_g \ket{\phi_g}
        &= \frac{1}{|G|^2}\sum_{g,h \in G}\bra{\phi_g} U_h^\dagger N_{g+h} U_h \ket{\phi_g} \\
        &= \frac{1}{|G|^2}\sum_{g,h \in G}\bra{\phi_{g+h}} N_{g+h} \ket{\phi_{g+h}} \\
        &= \frac{1}{|G|}\sum_{j \in G}\bra{\phi_j} N_j \ket{\phi_j}.
    \end{align*}
    Hence there exists an optimal POVM of the form $M_g = U_g M_0 U_g^\dagger$.

    Let $\mathcal{H}_0 \coloneqq \mathrm{span}\{U_g \ket{\phi} : g \in G\}$ be the orbit subspace generated by $\ket{\phi}$. 
    Because $G$ is finite abelian, its representation on $\mathcal{H}_0$ decomposes into one-dimensional irreps. Writing $\widehat{G}$ for the character group, define for each $\chi \in \widehat{G}$ the projector
    \begin{equation*}
        P_\chi \coloneqq \frac{1}{|G|}\sum_{g \in G} \overline{\chi(g)}\,U_g.
    \end{equation*}
    These satisfy
    \begin{equation*}
        U_h P_\chi = \chi(h) P_\chi,
        \qquad
        P_\chi P_{\chi'} = \delta_{\chi,\chi'} P_\chi.
    \end{equation*}
In particular, 
\[
U_h P_\chi 
= \frac{1}{\abs{G}}\sum_{g \in G} \overline{\chi(g)} \, U_{g + h}
= \frac{1}{\abs{G}}\sum_{k \in G} \overline{\chi(k - h)} \, U_{k}
= \chi(h) P_\chi,  
\]
where in the second-to-last equality we use the fact that $\chi$ is a group homomorphism, so $\chi(k-h) = \chi(k)\chi(-h)$, and that characters take values on the unit circle, hence $\overline{\chi(-h)}= \chi(h)$.  
Verifying the other identity, we have 
\begin{align*}
        P_\chi P_{\chi'}  
&= \frac{1}{\abs{G}^2} \sum_{g, h} \overline{\chi(g)\chi'(h)}U_{g+h} \\
&= \frac{1}{\abs{G}^2} \sum_{k, h} \overline{\chi(k-h)\chi'(h)}U_{k} \\
&= \frac{1}{\abs{G}^2} \sum_{k} \overline{\chi(k)}U_{k} \sum_h \chi(h)\overline{\chi'(h)} \\
&= \delta_{\chi,\chi'} P_\chi.
\end{align*}
The last equality follows from the orthogonality relations for characters of finite abelian groups, namely $\sum_{h \in G} \chi(h) \overline{\chi'(h)} = \abs{G}\delta_{\chi, \chi'}$.
One can similarly verify that $P_\chi^\dagger = P_\chi$, $\sum_\chi P_\chi = I$, and $P_\chi U_h = \chi(h) P_\chi$.
    Let $\ket{\phi_\chi} \coloneqq P_\chi \ket{\phi}$. (Recall our GU ensemble is defined by applying $U_g$ to $\ket{\phi}$ for each $g \in G$.)
    Then our earlier identities imply
    \begin{equation*}
        \ket{\phi} = \sum_{\chi \in \widehat{G}} \ket{\phi_\chi},
        \qquad
        U_g \ket{\phi_\chi} = \chi(g)\ket{\phi_\chi}.
    \end{equation*}
    So every vector in $P_\chi \calH_0$ is a scalar multiple of $\ket{\phi_\chi}$, and thus $P_\chi \calH_0$ is one-dimensional. 
    In particular, the nonzero $\ket{\phi_\chi}$ are mutually orthogonal. Therefore
    \begin{equation*}
        \ket{\phi_g} = U_g \ket{\phi} = \sum_{\chi \in \widehat{G}} \chi(g)\ket{\phi_\chi}.
    \end{equation*}
    Averaging over $g$ and using character orthogonality gives
    \begin{equation*}
        \rho
        = \frac{1}{|G|}\sum_{g \in G}\ketbra{\phi_g}{\phi_g}
        = \sum_{\chi \in \widehat{G}} \ketbra{\phi_\chi}{\phi_\chi}.
    \end{equation*}
    Since these rank-one terms have orthogonal supports, the nonzero eigenvalues of $\rho$ are exactly $\norm{\ket{\phi_\chi}}_2^2$. Hence
    \begin{equation*}
        \tr \sqrt{\rho} = \sum_{\chi \in \widehat{G}} \norm{\ket{\phi_\chi}}_2.
    \end{equation*}

    Next, let $\{M_g\}_{g \in G}$ be any POVM with $M_g = U_g M_0 U_g^\dagger$. Since $\sum_\chi P_\chi = I_{\calH_0}$ and $\sum_{g \in G} M_g = I_{\mathcal{H}_0}$ on $\mathcal{H}_0$, we have
    \begin{align*}
        I_{\mathcal{H}_0}
        &= \sum_{g \in G} U_g M_0 U_g^\dagger \\
        &= \sum_{g \in G} U_g \left(\sum_{\chi \in \widehat{G}} P_\chi\right) M_0 \left(\sum_{\chi' \in \widehat{G}} P_{\chi'}\right) U_g^\dagger \\
        &= \sum_{\chi,\chi' \in \widehat{G}}
           \left(\sum_{g \in G} \chi(g)\overline{\chi'(g)}\right)
           P_\chi M_0 P_{\chi'} \\
        &= |G| \sum_{\chi \in \widehat{G}} P_\chi M_0 P_\chi.
    \end{align*}
    Therefore, for every $\chi \in \widehat{G}$,
    \begin{align*}
        I_{\calH_0} = |G| \sum_{\chi' \in \widehat{G}} P_{\chi'} M_0 P_{\chi'} 
        \iff
        \frac{1}{\abs{G}} P_\chi I_{\calH_0} P_\chi = \sum_{\chi' \in \widehat{G}} P_{\chi}P_{\chi'} M_0 P_{\chi'} P_{\chi} 
        \iff 
        \frac{1}{\abs{G}} P_\chi = P_\chi M_0 P_\chi.
    \end{align*}
    In particular, we have
    \begin{equation*}
        \bra{\phi_\chi} M_0 \ket{\phi_\chi} 
        = \frac{1}{|G|} \bra{\phi_\chi} P_\chi \ket{\phi_\chi} 
        = \frac{1}{|G|}\norm{\ket{\phi_\chi}}_2^2.
    \end{equation*}

    The success probability of the POVM is
    \begin{equation*}
        p_{\mathrm{success}}
        = \frac{1}{|G|}\sum_{g \in G}\bra{\phi_g} M_g \ket{\phi_g}
        = \frac{1}{|G|}\sum_{g \in G}\bra{\phi}U^\dagger_g U_g M_0 U^\dagger_g U_g \ket{\phi}
        = \bra{\phi} M_0 \ket{\phi}
        = \sum_{\chi,\chi' \in \widehat{G}} \bra{\phi_\chi} M_0 \ket{\phi_{\chi'}}.
    \end{equation*}
    Since $M_0 \succeq 0$, we can apply Cauchy--Schwarz, so
    \begin{equation*}
        \abs{\bra{\phi_\chi} M_0 \ket{\phi_{\chi'}}}
        \leq \sqrt{\bra{\phi_\chi} M_0 \ket{\phi_\chi}\,\bra{\phi_{\chi'}} M_0 \ket{\phi_{\chi'}}}
        = \frac{1}{|G|}\norm{\ket{\phi_\chi}}_2 \norm{{\phi_{\chi'}}}_2.
    \end{equation*}
    Therefore,
\begin{align*}
p_{\mathrm{success}}
&= \sum_{\chi,\chi' \in \widehat{G}} \bra{\phi_\chi} M_0 \ket{\phi_{\chi'}} \\
&\leq \sum_{\chi,\chi' \in \widehat{G}} |\bra{\phi_\chi} M_0 \ket{\phi_{\chi'}}| \\ 
&\leq \frac{1}{|G|} \sum_{\chi,\chi' \in \widehat{G}} \norm{\ket{\phi_\chi}}_2 \norm{\ket{\phi_{\chi'}}}_2 \\ 
&= \frac{1}{|G|}\left(\sum_{\chi \in \widehat{G}} \norm{\ket{\phi_\chi}}_2\right)^2 \\
&= \frac{1}{|G|}\bigl(\tr \sqrt{\rho}\bigr)^2.
\end{align*}
    Thus no POVM can achieve success probability larger than $\frac{1}{|G|}(\tr \sqrt{\rho})^2$.

    Finally, consider the pretty-good measurement
    \begin{equation*}
        M_g^{\mathrm{PGM}}
        \coloneqq
        \rho^{-1/2}\,\frac{1}{|G|}\ketbra{\phi_g}{\phi_g}\,\rho^{-1/2}.
    \end{equation*}
    Since
    \begin{equation*}
        \rho^{-1/2}\ket{\phi}
        =
        \sum_{\chi:\,\phi_\chi \neq 0} \frac{\ket{\phi_\chi}}{\norm{\phi_\chi}_2},
    \end{equation*}
    we have 
    \begin{align*}
        p_{\mathrm{success}}^{\mathrm{PGM}}
        &= \bra{\phi} M_0^{\mathrm{PGM}} \ket{\phi} = \frac{1}{|G|}\left|\bra{\phi}\rho^{-1/2}\ket{\phi}\right|^2= \frac{1}{|G|}\left(\sum_{\chi \in \widehat{G}} \norm{\ket{\phi_\chi}}_2\right)^2= \frac{1}{|G|}\bigl(\tr \sqrt{\rho}\bigr)^2.
    \end{align*}
    Therefore, the pretty-good measurement attains the upper bound and is optimal. Noting that $\rho = \rho_{\mathrm{avg}}$ concludes the proof. 
\end{proof}

\section{\texorpdfstring{Random $\F_q$}{F_q}-phase states are close to Haar-random}\label{appendix:random-fq}

Let $\rho_{\rm Haar}^t$ denote the mixture of $t$ copies of a Haar-random pure state. Recall that  
\begin{equation*}
        \rho_{\rm all}^t \coloneqq \E_{\bc \sim \F_q^n}[\ketbra{\psi_{\bc}}{\psi_{\bc}}^{\otimes t}], 
\end{equation*}
where  $\ket{\psi_{\bc}} =\frac1{\sqrt n}\sum_{i=1}^n \omega^{{\bc}_i} \ket{i}$.

\randomfqphase*

Brakerski and Shmueli~\cite{brakerski10.1007/978-3-030-36030-6_10} proved the $\F_2$ case (see also \cite{ananth2022pseudorandom}). \cref{thm:random-fq-phase} extends this to general $\F_q$ and simplifies the argument.

\begin{proof}
If $t > n$, the bound holds trivially. Thus, we may assume $t \le n$.

The key observation is that the two ensembles agree exactly on the collision-free subspace; all disagreement is supported on tuples with repeated indices, whose total mass is bounded by $t(t-1)/n$.

Let $[n]^t$ index the computational basis of $(\mathbb{C}^n)^{\otimes t}$. For a basis string $a=(a_1,\dots,a_t)\in[n]^t$, we denote its occupation numbers by 
\[
m_a(j) \coloneqq \abs{\{\ell \in [t]: a_\ell=j\}}
\]
for $j \in [n]$. Let $\Pi_{\rm distinct}$ be the projector onto the collision-free subspace, which is the span of basis strings $a \in [n]^t$ with all $t$ entries distinct.

First, start by characterizing the matrix entries of $\rho_{\rm all}^t$. 
For any $a, b \in [n]^t$,
\begin{align*}
\langle a|\rho_{\rm all}^t|b\rangle 
&= \frac{1}{n^t}\mathbb{E}_{\bv \sim \mathbb{F}_q^n} \left[ \omega^{\sum_{\ell=1}^t (\bv_{a_\ell}-\bv_{b_\ell})} \right] \\
&= \frac{1}{n^t}\prod_{j=1}^n \mathbb{E}_{\bv_j \sim \mathbb{F}_q} \left[ \omega^{(m_a(j)-m_b(j))\bv_j} \right] \\
&= \frac{1}{n^t} \mathbf{1}\{m_a(j) \equiv m_b(j) \pmod q \text{ for all } j \in [n]\}.
\end{align*}

For $\rho_{\rm Haar}^t$, it is well-known~\cite{Har13} that we can express the state as a sum over permutations $\pi \in S_t$:
\[
\rho_{\rm Haar}^t = \frac{\Pi_{\rm sym}}{\dim \mathrm{Sym}^t(\mathbb{C}^n)} = \frac{1}{n(n+1)\cdots(n+t-1)} \sum_{\pi\in S_t} U_\pi,
\]
where $U_\pi$ acts by permuting the $t$ copies, $\mathrm{Sym}^t(\C^n)$ denotes the symmetric subspace, and $\Pi_{\rm sym}$ is the projector onto the symmetric subspace. 
Thus, its matrix entries are
\[
\langle a|\rho_{\rm Haar}^t|b\rangle = \frac{1}{n(n+1)\cdots(n+t-1)} |\{\pi \in S_t : a = b \circ \pi\}|.
\]

Now we restrict both states to the collision-free subspace. 
If $a$ and $b$ are both collision-free, then every occupation number $m_a(j)$ and $m_b(j)$ lies in $\{0,1\}$. Because $q \ge 2$, equivalence modulo $q$ is identical to strict equality:
\[
m_a(j) \equiv m_b(j) \pmod q \iff m_a(j) = m_b(j).
\]
Because $m_a(j)=m_b(j)$ for all $j$ exactly when $a$ and $b$ are permutations of one another, the two ensembles have identical support on the collision-free subspace.
Moreover, since all entries in a collision-free string are distinct, there is exactly one permutation matching $b$ to $a$. 

Explicitly, let us define the unnormalized block-diagonal operator
\[
A \coloneqq \sum_{\substack{a,b \in [n]^t \\ a,b \text{ collision-free} \\ \{a_1,\dots,a_t\} = \{b_1,\dots,b_t\}}} |a\rangle\langle b|.
\]
Then the restrictions of the two ensembles to the collision-free subspace are purely proportional to $A$:
\[
\Pi_{\rm distinct} \rho_{\rm all}^t \Pi_{\rm distinct} = 
\frac{1}{n^t}A, \qquad \Pi_{\rm distinct} \rho_{\rm Haar}^t \Pi_{\rm distinct} = \frac{1}{n(n+1)\cdots(n+t-1)}A.
\]
Let $(n)_t \coloneqq n(n-1)\cdots(n-t+1)$ denote the falling factorial. The trace of $A$ is exactly the number of collision-free strings in $[n]^t$, so $\tr(A) = (n)_t$. Hence, the traces of the two restricted states are
\[
p_{\rm all} \coloneqq \tr(\Pi_{\rm distinct} \rho_{all}^t \Pi_{\rm distinct}) = \frac{(n)_t}{n^t}, \qquad p_{\rm Haar} \coloneqq \tr(\Pi_{\rm distinct} \rho_{\rm Haar}^t \Pi_{\rm distinct}) = \frac{(n)_t}{n(n+1)\cdots(n+t-1)}.
\]
After normalization, both restricted states are identical to the same conditional state $\tau \coloneqq A / (n)_t$.

Furthermore, both ensembles are block diagonal with respect to the decomposition $(\mathbb C^n)^{\otimes t} = \Pi_{\rm distinct}\oplus(I-\Pi_{\rm distinct}).$
For the Haar ensemble, this is immediate since permutations preserve occupation numbers. For the phase ensemble, any nonzero matrix entry satisfies $m_a(j)\equiv m_b(j)\pmod q$ for all $j\in[n]$. In particular, if $a$ is collision-free, then each occupation number is either $0$ or $1$, and since $\sum_j m_b(j)=t$, the same must hold for $b$. Hence $b$ is also collision-free.
Therefore, we can write 
\begin{align*}
\rho_{\rm all}^t &= p_{\rm all} \tau + (1-p_{\rm all}) \rho_{\rm bad}, \\
\rho_{\rm Haar}^t &= p_{\rm Haar} \tau + (1-p_{\rm Haar}) \sigma_{\rm bad},
\end{align*}
where the states $\rho_{\rm bad}$ and $\sigma_{\rm bad}$ are valid density matrices fully supported on $I-\Pi_{\rm distinct}$. 

By the triangle inequality:
\begin{align*}
\mathrm{D}_{\rm tr}(\rho_{\rm all}^t, \rho_{\rm Haar}^t) 
&= \frac{1}{2} \norm{(p_{\rm all} - p_{\rm Haar}) \tau + (1-p_{\rm all}) \rho_{\rm bad} - (1-p_{\rm Haar}) \sigma_{\rm bad}}_1 \\
&\le \frac{1}{2} \norm{(p_{\rm all} - p_{\rm Haar}) \tau}_1 + \frac{1}{2} \norm{(1-p_{\rm all}) \rho_{\rm bad}}_1 + \frac{1}{2} \norm{ (1-p_{\rm Haar}) \sigma_{\rm bad}}_1 \\
&= \frac{1}{2} \left( |p_{\rm all} - p_{\rm Haar}| + (1-p_{\rm all}) + (1-p_{\rm Haar}) \right).
\end{align*}
Notice that $n^t \le n(n+1)\cdots(n+t-1)$, and therefore $p_{\rm Haar} \le p_{\rm all} \le 1$. 
Thus, $|p_{\rm all} - p_{\rm Haar}| = p_{\rm all} - p_{\rm Haar}$, and the entire bound simplifies to
\[
\mathrm{D}_{\rm tr}(\rho_{\rm all}^t, \rho_{\rm Haar}^t) \le 1 - p_{\rm Haar}.
\]
It remains to bound $1 - p_{\rm Haar}$. 
Recall that 
\[
p_{\rm Haar} = \frac{(n)_t}{n(n+1)\dots(n+t -1)} = \prod_{j=0}^{t-1} \frac{n-j}{n+j} = \prod_{j=0}^{t-1} \left(1 - \frac{2j}{n+j}\right).
\]
Since $t \le n$, we have $\frac{2j}{n+j} \le 1$ for all $0 \le j \le t-1$. Using the bound $1 - \prod_j (1-x_j) \le \sum_j x_j$ for $x_j \in [0, 1]$, we obtain
\[
1 - p_{\rm Haar} \le \sum_{j=0}^{t-1} \frac{2j}{n+j} \le \sum_{j=0}^{t-1} \frac{2j}{n} = \frac{t(t-1)}{n}.
\]
This completes the proof.
\end{proof}

\end{document}